\documentclass[journal,twoside,web]{ieeecolor}
\usepackage{generic}
\usepackage{cite}
\usepackage{amsmath,amssymb,amsfonts}
\usepackage{graphicx}
\usepackage{hyperref}
\usepackage{algpseudocode}
\usepackage{textcomp}
\usepackage{xcolor}
\usepackage{stfloats}
\usepackage{booktabs}
\let\labelindent\relax
\usepackage{enumitem}
\usepackage{arydshln}
\usepackage{textcomp}
\usepackage{multirow}
\usepackage{bm}

\hypersetup{hypertex=true,
	colorlinks=true,
	linkcolor=blue,
	urlcolor=black,
	anchorcolor=blue,
	citecolor=blue}

\newtheorem{assumption}{Assumption}[section]
\newtheorem{definition}{Definition}[section]

\newtheorem{theorem}{Theorem}[section]

\newtheorem{lemma}{Lemma}

\newtheorem{remark}{Remark}

\newtheorem{problem}{Problem}

\newcommand\norm[1]{\left\lVert#1\right\rVert}

\def\BibTeX{{\rm B\kern-.05em{\sc i\kern-.025em b}\kern-.08em
    T\kern-.1667em\lower.7ex\hbox{E}\kern-.125emX}}
\begin{document}
\title{Finite Sample Analysis of Open-loop Subspace Identification Methods}
\author{Jiabao He, Ingvar Ziemann, Cristian R. Rojas, S. Joe Qin and H\r{a}kan Hjalmarsson
\thanks{Jiabao He, Cristian R. Rojas and H\r{a}kan Hjalmarsson are with the Division of Decision and Control Systems, School of Electrical Engineering and Computer Science, KTH Royal Institute of Technology, 100 44 Stockholm, Sweden.  (Emails: jiabaoh, crro, hjalmars@kth.se)}
\thanks{Ingvar Ziemann is with the University of Pennsylvania, Philadelphia, PA 19104 USA. (Email: ingvarz@seas.upenn.edu)}
\thanks{S. Joe Qin is with the Institute of Data Science, Lingnan University, Hong
Kong (Email: joeqin@ln.edu.hk)}
\thanks{This work was supported by VINNOVA Competence Center AdBIOPRO, contract [2016-05181] and by the Swedish Research Council through the research environment NewLEADS (New Directions in Learning Dynamical Systems), contract [2016-06079], and contract 2019-04956.}}

\maketitle
\begin{abstract}
Subspace identification methods (SIMs) are known for their simple parameterization for MIMO systems and robust numerical properties. However, a comprehensive statistical analysis of SIMs remains an open problem. Following a three-step procedure generally used in SIMs, this work presents a finite sample analysis for open-loop SIMs. In Step 1 we begin with a parsimonious SIM. Leveraging a recent analysis of an individual ARX model, we obtain a union error bound for a Hankel-like matrix constructed from a bank of ARX models. Step 2 involves model reduction via weighted singular value decomposition (SVD), where we use robustness results for SVD to obtain error bounds on extended controllability and observability matrices, respectively. The final Step 3 focuses on deriving error bounds for system matrices, where two different realization algorithms, the MOESP type and the CVA type, are studied. Our results not only agree with classical asymptotic results, but also show how much data is needed to guarantee a desired error bound with high probability. The proposed method generalizes related finite sample analyses and applies broadly to many variants of SIMs.
\end{abstract}

\vspace{-3mm}
\section{Introduction} \label{Sct1}
Originating from the celebrated Ho-Kalman algorithm \cite{Ho1966effective}, subspace identification methods (SIMs) have proven extremely useful for estimating linear state-space models and became one of the mainstream approaches in system identification. Over the past 50 years, numerous efforts have been made to develop improved algorithms and gain a deeper understanding of them. For a comprehensive overview of SIMs, we refer to \cite{Qin2006overview,van2013closed}. Overall speaking, SIMs can be categorized into two types, namely, the open-loop and closed-loop. Open-loop SIMs were developed first and formed the basis for the development of closed-loop ones. Some representative open-loop SIMs are canonical variate analysis (CVA) \cite{Larimore1990canonical}, numerical algorithms for subspace state-space system identification (N4SID) \cite{Van1994n4sid}, multivariable output-error state-space (MOESP) algorithms \cite{Verahegen1992subspace}, the observer-Kalman filter method (OKID) \cite{Phan1995improvement}, and the parsimonious SIM (PARSIM) \cite{Qin2005novel,He2024weighted}. Despite the significant theoretical and practical success of SIMs, certain limitations remain--most notably, their lower accuracy compared to the prediction error method (PEM) in the case of exogenous inputs being present, and the lack of a comprehensive statistical analysis. A thorough statistical analysis of SIMs is essential to establish their reliability, assess their performance, and guide the design of more robust and efficient algorithms and the choice of user choices.

\vspace{-3mm}
\subsection{Related Work}  \label{Sct1.1}
There are some significant contributions to statistical properties of SIMs in the asymptotic regime \cite{Deistler1995consistency,Larimore1996statistical,Peternell1996statistical,Jansson1998consistency,Bauer1999consistency,Knudsen2001consistency,Bauer2000analysis,Jansson2000asymptotic,Gustafsson2002subspace,Bauer2005asymptotic,Chiuso2004asymptotic,Chiuso2005consistency,Chiuso2007relation,Chiuso2007role}. The consistency and asymptotic variance of SIMs are analyzed in \cite{Jansson1998consistency,Chiuso2005consistency} and \cite{Bauer2000analysis,Jansson2000asymptotic,Gustafsson2002subspace,Bauer2005asymptotic,Chiuso2004asymptotic}, respectively. In particular, it was pointed out in \cite{Jansson1998consistency} that the persistence of excitation (PE) of inputs is not sufficient for consistency, and stronger conditions are required in some cases. In addition, the impact of some weighting matrices in the SVD step was discussed in \cite{Bauer2000impact,Bauer2002some}, which claim that the choices of weighting matrices mainly influence the asymptotic distribution of the estimates. Although the CVA method gives the lowest variance among available weighting choices when the measured inputs are white \cite{Larimore1996statistical}, there is no formal proof to show that it is asymptotically efficient \cite{Chiuso2007role}. In the asymptotic regime, convergence rates can be derived using the central limit theorem (CLT) or the law of the iterated logarithm (LIL) \cite{Hannan2012statistical}. However, such results only hold as the number of samples tends to infinity. In reality, all data is finite. Asymptotic results serve primarily as heuristics and often fall short in capturing transient behaviors, explaining performance differences among SIM variants in finite sample settings, or determining how much data is needed to get a model with which we are satisfied.

There has been a recent resurgence of interest in identifying state-space models, where the focus is on the non-asymptotic regime. Finite sample analysis in the field of system identification was pioneered by \cite{Campi2002finite,Weyer1999finite}, where the performance of PEMs was analyzed. Over the last few years, a series of papers have revisited this topic and introduced many promising developments on fully observed systems \cite{Sarkar2019near,Simchowitz2018learning,Jedra2022finite} and partially observed systems \cite{Tsiamis2019finite,Oymak2019non,Simchowitz2019learning,Sarkar2021finite,Lale2021finite,Bakshi2023new}. For a broader overview of these results, we refer to \cite{Tsiamis2023statistical}. As stated in \cite{Tsiamis2019finite}, finite sample analysis has been a standard tool for comparing algorithms. Such an analysis of SIMs can provide a qualitative characterization of learning complexity and elucidate data-accuracy trade-offs. Moreover, it can guide the choice of user choices, such as the horizons and weighting matrices, which may lead to an improvement of the estimator in a two-step approach. However, the path of finite sample analysis for SIMs proves to be challenging due to the involvement of multi-step statistical operations \cite{Bauer2005asymptotic}, such as regression, projection, weighted SVD and maximum likelihood (ML) estimation. While these steps enhance performance, they simultaneously pose challenges for any subsequent statistical analysis. Putting the studies on fully observed systems aside, the most relevant studies on partially observed systems are \cite{Oymak2019non,Sun2023finite,Tsiamis2019finite,Lale2021finite}. However, they mainly analyze the performance of the Ho-Kalman algorithm or similar variants which are rarely used in practice, their results therefore are not sufficient to completely reveal statistical properties of SIMs actually deployed. To the best of our knowledge, a complete finite sample analysis of SIMs under general conditions is still an open problem.
% The first reason is that the Ho-Kalman algorithm is rarely used in the literature of SIMs. For instance, a key step in SIMs is the weighted SVD, where some data-dependent weighting matrices are pre-multiplied and post-multiplied to a Hankel matrix before performing an SVD. The selection of weighting matrices turns out to be crucial for improving the performance of SIMs. Beyond the trivial identity matrix, the impact of other data-dependent weighting matrices has not been considered. The second reason is that unlike the Ho-Kalman algorithm, many SIMs typically estimate system matrices by first recovering the state sequences and then applying least-squares regression in the output and state equations. To the best of our knowledge, this realization algorithm has not been analyzed in the non-asymptotic regime. Moreover, some studies, such as \cite{Tsiamis2019finite}, omit input signals which would pose additional challenges, particularly if a projection step is involved. In short, the-state-of-art in finite sample analysis of SIMs streamlines the realization steps, and a complete finite sample analysis under general conditions is still an open problem.
\vspace{-3mm}
\subsection{Contributions}  \label{Sct1.2}

The main contributions of this paper are three-fold: 

(1) We develop a robust and scalable framework for finite sample analysis of a broad class of SIMs. To avoid non-causal models caused by the projection step in classical SIMs, we propose to use PARSIM to enforce a causal model. Such a choice brings convenience to statistical analysis, and the method can be applied to other ARX-based SIMs, such as SSARX \cite{Jansson2003subspace} and PBSID \cite{Chiuso2005consistency}.

(2) We establish a more general PE condition. Compared with related studies that only include past inputs and past outputs as regressors, our work also includes future inputs as regressors, leading to a more general PE condition. This broader PE condition is instrumental in deriving error bounds and in analyzing the use of data-dependent weighting matrices. Therefore, it serves as a contribution of independent interest.

(3) Compared with related studies that streamline the realization algorithm, we provide the first finite sample upper bounds on system matrices under different weighting matrices and two popular realization algorithms. 

A preliminary version \cite{He2024finite} of this work was accepted by IEEE CDC24, where we provide a finite sample analysis for a simplified PARSIM, i.e., without taking into account the weighting matrices and including the CVA type realization algorithm. In this full version, we include different weighting matrices and two popular realization algorithms. In addition, we also provide complete proofs and a technical framework to approach this problem.
\vspace{-3mm}
\subsection{Structure}  \label{Sct1.3}
The disposition of the paper is as follows: In Section~\ref{Sct2}, we introduce models and assumptions used in SIMs, and then formulate the problem explicitly. In Section~\ref{Sct3}, we present a short review of SIMs with a focus on PARSIM, and the roadmap ahead to analyze its finite sample behavior. In Section~\ref{Sct4}, we first provide a finite sample analysis of an individual ARX model, which we then combine with a union bound to control the performance of a bank of ARX models. In Section~\ref{Sct5}, we first analyze certain robustness properties of weighted SVD, and then derive error bounds on the system matrices coming from two realization algorithms. In Section~\ref{Sct6}, we discuss the implications of our main results. Finally, the paper is concluded in Section~\ref{Sct7}. All proofs and technical lemmas are provided in the Appendix.
\vspace{-3mm}
\subsection{Notations}  \label{Sct1.4}

(1) For a matrix $X$ with appropriate dimensions, $X^\top$, $X^{-1}$, $X^{\frac{1}{2}}$, $X^\dagger$, $\lVert X \rVert$, $\lVert X \rVert_F$, ${\rm{det}}(X)$, ${\rm{rank}}(X)$, ${\rm{trace}}(X)$, $\rho(X)$, $\lambda_{\rm{max}}(X)$, $\lambda_{\rm{min}}(X)$, $\sigma_{\rm{min}}(X)$ and $\sigma_{n}(X)$ denote its transpose, inverse, square root, Moore$\mbox{-}$Penrose pseudo-inverse, spectral norm, Frobenius norm, determinant, rank, trace, spectral radius, maximum eigenvalue, minimum eigenvalue, minimum singular value and $n$-th largest singular value, respectively. Moreover, $X_1 \succ(\succcurlyeq)$ $0$ and $X_2 \prec(\preccurlyeq)$ $0$ mean that $X_1$ is positive (semi) definite and $X_2$ is negative (semi) definite, respectively. ${\rm{diag}}(X_1,X_2)$ is a block matrix having $X_1$ and $X_2$ on its diagonal. The matrices $I$ and $0$ are the identity and zero matrices with compatible dimensions. 

(s) The multivariate normal distribution with mean $\mu$ and covariance $\Sigma$ is denoted as $\mathcal{N}(\mu,\Sigma)$. The notation $\mathbb{E}\left[x\right]$ is the expectation of a random vector $x$. For an event $\mathcal{E}$, $\mathbb{P}(\mathcal{E})$ is the probability of $\mathcal{E}$, $\mathcal{E}^c$ is the complementary event of $\mathcal{E}$, and $\mathcal{E}_{1}\cup\mathcal{E}_{2}$ and $\mathcal{E}_{1}\cap\mathcal{E}_{2}$ are the union and intersection of events $\mathcal{E}_{1}$ and $\mathcal{E}_{2}$, respectively. We use $\mathbb{I}_{\left\{ \mathcal{E} \right\}}$ to denote the indicator function of $\mathcal{E}$.

(3) The notation $f = \mathcal{O}(g)$ means that functions $f,g \in \mathbb{R}^{d}$ satisfy $\limsup_{x\to x_0}{|\frac{f(x)}{g(x)}|}< \infty$, where the limit point $x_0$ is typically understood from the context. Moreover, $f \gtrsim g$ means $f$ is greater than or approximately equal to $g$.

(4) The notations $c$, $c_1$, ... stand for universal constants independent of system parameters, confidence, and accuracy. 
\vspace{-3mm}
\section{Problem Formulation} \label{Sct2}

\subsection{Models and Assumptions} \label{Sct2.1}

Consider the following discrete-time linear time-invariant (LTI) system in innovations form:
\begin{subequations} \label{E1}
	\begin{align}
		x_{k + 1} &= Ax_{k}  + Bu_{k} + K{\Sigma_e^{\frac{1}{2}}}e_{k}, \label{E1a}\\
		y_{k} &= Cx_{k} + {\Sigma_e^{\frac{1}{2}}}e_{k}, \label{E1b}		
	\end{align}
\end{subequations}
where $x_{k}\in \mathbb{R}^{n_x}$, $u_{k}\in \mathbb{R}^{n_u}$, $y_{k}\in \mathbb{R}^{n_y}$ and $e_{k}\in \mathbb{R}^{n_y}$ are the state, input, output and innovations, respectively. For brevity of notation, we assume that the initial time starts at $k=1$, and the terminal time is denoted as $\bar N = N+p+f-1$, where $N$ is the number of columns in data Hankel matrices, and $p$ and $f$ stand for past and future horizons, respectively, to be defined later. In addition, the initial state is assumed to be $x_1=0$. We make the following standard assumptions:
\begin{assumption} \label{Asp1}
	(1) The spectral radius of $A$ and $A_K$ satisfy $\rho({A}) < 1$ and $\rho(A_K) < 1$.
	 
	(2) The system is minimal, i.e., $(A,[B,K])$ is controllable and $(A,C)$ is observable.
	
	(3) The innovations $\{e_k\}$ consists of independent and identically distributed (i.i.d.) Gaussian random variables, i.e., $e_{k} \sim \mathcal{N}(0,I)$.\footnote{Similar to \cite{Ziemann2023tutorial}, our results can be extended to more general setups, such as sub-Gaussians.}
	
	(4) The input sequence $\{u_k\}$ consists also of i.i.d. Gaussian random variables, i.e., $u_{k} \sim \mathcal{N}(0,\sigma_u^2I)$. Moreover, it is assumed independent of $\{e_k\}$.
\end{assumption}

\begin{remark} \label{Rmk1minus}
To illustrate the generality of the innovations model \eqref{E1}, we consider the following standard state-space model which divides the noise term into contributions from measurement noise $v_k$ acting on the outputs and process noise $w_k$ acting on the states \cite{Ljung1999system}:
\begin{subequations} \label{E1B}
	\begin{align}
		x_{k + 1} &= Ax_{k}  + Bu_{k} + w_{k}, \label{E1Ba}\\
		y_{k} &= Cx_{k} +v_{k}. \label{E1Bb}		
	\end{align}
\end{subequations}
The noises $w_k$ and $v_k$ consist of i.i.d. zero-mean Gaussian random variables, with covariance $\Sigma_w$ and $\Sigma_v$, respectively. Moreover, they are independent of each other. We assume that $\Sigma_v \succ 0$, $(A,C)$ is detectable, and $(A,\Sigma_w)$ is stabilizable. Then, the Kalman filter of system \eqref{E1B} is well defined, and the Kalman gain is equal to 
\begin{equation*}
      K = -APC^\top\left(CPC^\top + \Sigma_v\right)^{-1},
\end{equation*}
where $P$ is the solution of the following Riccati equation:
\begin{equation*}
      P = APA^\top + \Sigma_w - APC^\top\left(CPC^\top + \Sigma_v\right)^{-1}\left(APC^\top\right)^\top.
\end{equation*}
We further assume that the initial state is a zero-mean Gaussian variable with covariance $P$ and independent of the noises. Then, by the orthogonality principle the innovations sequence also consists of i.i.d. Gaussian variables with covariance 
\begin{equation*}
	\Sigma_e = CPC^\top + \Sigma_v.
\end{equation*}
Therefore, under mild conditions the innovations form \eqref{E1} describes the same input-output trajectories as the standard state-space model \eqref{E1B}, and it is widely used in SIMs \cite{Qin2006overview}.
\end{remark}

Based on the innovations form, after replacing $\Sigma_e^{1/2}e_{k}$ in \eqref{E1a} with $y_{k} - Cx_{k}$, we obtain the following predictor form:
\begin{subequations} \label{E1A}
	\begin{align}
		x_{k + 1} &= A_Kx_{k}  + Bu_{k} + Ky_{k}, \label{E1Aa}\\
		y_{k} &= Cx_{k} + \Sigma_e^{1/2}e_{k}, \label{E1Ab}		
	\end{align}
\end{subequations}
where $A_K = A-KC$. Since the innovations form and the predictor form are equivalent and all can represent input and output data exactly, one has the option to use any of these forms for convenience. For instance, MOESP \cite{Verahegen1992subspace} and PARSIM \cite{Qin2005novel} use the innovations form, and SSARX \cite{Jansson2003subspace} and PBSID \cite{Chiuso2007role} use the predictor form.
\vspace{-3mm}
\subsection{Problem Formulation} \label{Sct2.2}

In this paper we tackle the following problem:
\begin{problem} \label{Prom1}
	Given a finite number $\bar N$ of input-output samples from a single trajectory of system \eqref{E1}, our goal is to explicitly derive high probability error bounds on the system matrices estimated by some SIMs. To be specific, given a confidence level $0<\delta<1$, we wish to derive error bounds $\epsilon_A, \epsilon_B, \epsilon_C$, such that
	\begin{equation*} \label{E22}
		\norm{\hat A - T^{-1}AT} \leq \epsilon_A, \norm{\hat B - T^{-1}B} \leq \epsilon_B, \norm{\hat C - CT} \leq \epsilon_C,
	\end{equation*}
	hold with probability at least $1-\delta$, where $\hat A$, $\hat B$ and $\hat C$ are estimates of system matrices, and $T$ is a non-singular matrix.
\end{problem}
    
\begin{remark}
	It is only possible to obtain the system matrices up to a similarity transformation due to the non-uniqueness of a realization\cite{Oymak2019non}. Moreover, it should be mentioned that the matrix $T$ here is stochastic, depending on the realization. Alternatively the estimates could be transformed into a canonical form. Furthermore, bounds on the estimation of Markov parameters or other system invariants could also be given.
\end{remark}

Problem \ref{Prom1} is one of the long-standing open problems in subspace identification. As noted by Van Overschee and De Moor \cite{Van2012subspace}, ``solving these problems would contribute significantly to the maturing of the field of subspace identification.'' The existing literature already offers some pertinent solutions to Problem \ref{Prom1}. According to recent work \cite{Tsiamis2019finite,Ziemann2023tutorial} in the non-asymptotic regime, the error bound $\epsilon_\theta$ is typically of the form
\begin{equation} \label{error-bound}
	\epsilon_\theta \propto (\text{SNR})^{-1} \times\sqrt{\frac{\text{problem dimension} + \text{log}(1/\delta)}{\bar N}},
\end{equation}
where $\theta$ denotes the parameter of interest, and $\text{SNR}$ denotes the signal-to-noise ratio. 

In the asymptotic regime, prior work \cite{Bauer2000analysis,Jansson2000asymptotic,Gustafsson2002subspace,Bauer2005asymptotic,Chiuso2004asymptotic} have shown that the normalized error $\sqrt{\bar N}\left(\hat \theta - \theta\right)$ converges in law to a normal distribution. Consequently, the results in \eqref{error-bound}--where the error decays at rate $\mathcal{O}(1/\sqrt{\bar N})$ and the confidence level $\delta$ appears through $\log(1/\delta)$--are consistent with these asymptotic results. Moreover, the LIL suggests that error decays at rate $\mathcal{O}(\sqrt{\frac{\log\log \bar N}{\bar N}})$ almost surely \cite{Deistler1995consistency}, which is sharp. However, asymptotic results require that $\bar N \to \infty$ and can only be used as a heuristic for a finite $\bar N$. Some questions remained unanswered. For instance, there often is a minimal requirement on $\bar N$, namely the burn-in time $\bar N_{\text{pe}}$, which is necessary for a bound of the form \eqref{error-bound} to hold. Such requirements are typically of the form
\begin{equation} \label{burn-in}
	\bar N_{\text{pe}} \gtrsim \text{problem dimension} + \text{log}(1/\delta).
\end{equation}
The above $\bar N_{\text{pe}}$ cannot be obtained by applying only asymptotic tools \cite{Tsiamis2023statistical}. Moreover, as shown in \cite{Tsiamis2019finite} and Lemma \ref{Lem1} in this work, non-asymptotic analysis can even achieve an error bound for marginally stable systems, whereas asymptotic results are often limited to asymptotically stable systems.

Existing non-asymptotic results for Problem \ref{Prom1} mainly focus on the classical Ho--Kalman algorithm or similar variants \cite{Oymak2019non,Tsiamis2019finite,Lale2021finite,Sarkar2021finite}. However, this realization algorithm is outdated, as more powerful SIMs have later been proposed in the literature. Whether the structure in \eqref{error-bound} also holds for modern SIMs has therefore remained unclear. This work proves that modern SIMs also obey the same structure. Moreover, it is different in the following key respects.

First, a standard step of SIMs is to estimate a Hankel (or Hankel-like) matrix of Markov parameters, denoted by $\mathcal{H}_{fp}$. A key feature of modern SIMs is that $\mathcal{H}_{fp}$ is estimated directly using a projection method. However, this projection step couples future data with past data, resulting in an error without a standard martingale structure, which is difficult to upper bound. We believe that this is one of the main barriers preventing a finite sample analysis for SIMs. Previous analyses either revert to the Ho--Kalman algorithm \cite{Oymak2021revisiting} or avoid inputs \cite{Tsiamis2019finite}, where the projection step is not involved, thereby sidestepping the problem. The way we solve it is to absorb the projection matrix into an enlarged regressor, so that the error restores a standard martingale structure. The price to pay is that this format requires a more complex yet tractable PE condition.

Second, in the model reduction step, due to the fact that $\mathcal{H}_{fp}$ is low-rank, some data-dependent weighting matrices $W_1$ and $W_2$ are pre-multiplied and post-multiplied to the estimate of ${\mathcal{H}}_{fp}$ before performing an SVD to improve the numerical and statistical properties. Several asymptotic properties of such algorithmic variations have been studied in \cite{Bauer2005asymptotic}. However, it is an open problem to study the impact of weighting matrices and compare their performance in a finite sample setting \cite{Tsiamis2023statistical,Van2012subspace}. Prior work \cite{Tsiamis2019finite,Oymak2021revisiting,Sarkar2021finite,Lale2021finite} considered the trivial weighting $W_1=W_2=I$. This work delivers the first finite sample analysis for general weighting matrices through a novel analysis leveraging the Schur complement.

Third, unlike the Ho-Kalman algorithm, many SIMs typically estimate system matrices by first recovering the state sequences and then applying least-squares regression in the output and state equations. To the best of our knowledge, this realization algorithm has not been analyzed in the non-asymptotic regime. Our work shows that the error of this realization algorithm also obeys the structure of \eqref{error-bound}.

In summary, our work extends finite sample results from simplified prototypes to the algorithms actually deployed, offering new insights and systematic performance guarantees. 

\vspace{-3mm}
\section{A Recap of Subspace Identification Methods} \label{Sct3}

In this section we provide a short overview of open-loop SIMs, with the focus on PARSIM. For convenience, we define 
\begin{equation*} \label{E24a}
	\begin{split}
		u_{p}(k) &= \begin{bmatrix}{{u_{k}^\top}}&{{u_{k+1}^\top}}&\cdots&{{u_{k+p-1}^\top}}\end{bmatrix}^\top \in \mathbb{R}^{pn_u}, \\
		u_{f}(k) &=	\begin{bmatrix}{{u_{k+p}^\top}}&{{u_{k+p+1}^\top}}&\cdots&{{u_{k+p+f-1}^\top}} \end{bmatrix}^\top \in \mathbb{R}^{fn_u},
	\end{split}	
\end{equation*}
which stack past inputs and future inputs, respectively. Similar definitions apply to $y_{p}(k)$, $e_{f}(k)$, $u_{i}(k)$ and $e_{i}(k)$. Moreover, after lining up $u_{p}(k)$ and $u_{f}(k)$ from $k=1$ to $k=N$, we obtain Hankel matrices
\begin{subequations} \label{E5}
	\begin{align}
		U_{p} &= \begin{bmatrix}
			u_{p}(1)&u_{p}(2)& \cdots &u_{p}(N) \end{bmatrix}, \label{E5a} \\
		U_{f} &= \begin{bmatrix}
			u_{f}(1)&u_{f}(2)& \cdots &u_{f}(N) \end{bmatrix}. \label{E5c}
	\end{align}
\end{subequations}
Similar definitions apply to data matrices $Y_p$, $Y_f$, $E_p$ and $E_f$ \cite{Qin2005novel}. The state sequence is given by
\begin{equation} \label{E5}
	X_{k} = \begin{bmatrix}{{x_k}}&{{x_{k+1}}}& \cdots &{{x_{k+N-1}}} \end{bmatrix}.
\end{equation}

By iterating model \eqref{E1} using the above notations, an extended state-space model \cite{Qin2005novel} can be derived as
\begin{subequations} \label{E2}
	\begin{align}
		Y_f &= \Gamma_fX_k + G_fU_f + H_fE_f, \label{E2a}\\
		Y_p &= \Gamma_pX_{k-p} + G_pU_p + H_pE_p, \label{E2b}		
	\end{align}
\end{subequations}
where $\Gamma_f$ is the extended observability matrix, defined by
\begin{equation} \label{E3}
	\Gamma_f = \begin{bmatrix}
		{{C^{\top}}}&{{{\left( {CA} \right)}^{\top}}}& \cdots &{{{\left({C{A^{f - 1}}} \right)}^{\top}}}
	\end{bmatrix}^{\top}.
\end{equation}
Moreover, the transmission matrices $G_f$ is a lower-triangular Toeplitz matrix of Markov parameters, 
\begin{equation} \label{E4}
	G_{f} = \begin{bmatrix}
		{0}&{0}& \cdots &0\\
		{CB}&0& \cdots &0\\
		\vdots & \vdots & \ddots & \vdots \\
		{C{A^{f-2}}B}&{C{A^{f-3}}B}& \cdots &0
	\end{bmatrix},
\end{equation}
and $H_{f}$ is similarly defined by replacing $0$ on the diagonal of $G_{f}$ with $\Sigma_e^{\frac{1}{2}}$, and by replacing $B$ with $K\Sigma_e^{\frac{1}{2}}$. Similar definitions apply to matrices $\Gamma_p$, $G_p$ and $H_p$. Furthermore, by iterating equation \eqref{E1Aa}, we obtain
\begin{equation} \label{E7}
	X_k = L_pZ_p + A_K^pX_{k-p},
\end{equation}
where $Z_p = \begin{bmatrix}
	Y_p^{\top}&U_p^{\top}
\end{bmatrix}^{\top}$ and $L_p$ is the extended controllability matrix in a reverse order, defined by
\begin{equation} \label{E8}
	L_p = \begin{bmatrix}
		A_K^{p-1}K&\cdots&K&A_K^{p-1}B&\cdots&B
	\end{bmatrix}.
\end{equation}
After substituting \eqref{E7} into \eqref{E2a}, we have 
\begin{equation} \label{E9}
	Y_f = \mathcal{H}_{fp}Z_p + G_fU_f + H_fE_f + \Gamma_f A_K^pX_{k-p},
\end{equation}
where $\mathcal{H}_{fp} := \Gamma_fL_p$ is the Hankel-like matrix. Most variants of SIMs can be integrated into a unified framework \cite{Van1995unifying,Qin2006overview} which generally consists of three steps. We now introduce them based on \eqref{E9}.
\vspace{-3mm}
\subsection{Step 1: Linear Regression or Projection} \label{Sct3.1}
Most open-loop SIMs use \eqref{E9} to first estimate the Hankel-like matrix $\mathcal{H}_{fp}$, and then proceed to obtain the system matrices. A basic approach in classical SIMs is one-step regression \cite{Verahegen1992subspace,Van1994n4sid,Jansson1998consistency,Knudsen2001consistency}, which takes $Z_p$ and $U_f$ as regressors and obtains $\mathcal{H}_{fp}$ and $G_f$ simultaneously using
\begin{equation} \label{E10}
	\hat \Theta := \begin{bmatrix}
		\hat{\mathcal{H}}_{fp}& \hat G_{f}
	\end{bmatrix} =  Y_f\begin{bmatrix}
		Z_p\\ U_f\end{bmatrix}^\dagger.
\end{equation}
Since $\mathcal{H}_{fp}$ is our main interest, using the inverse of a block matrix (see Lemma \ref{LemA8}), $\hat{\mathcal{H}}_{fp}$ can be extracted from \eqref{E10} as
\begin{equation} \label{E11}
	\hat{\mathcal{H}}_{fp} = Y_f\Pi_{U_f}^{\perp}Z_p^\top(Z_p\Pi_{U_f}^{\perp}Z_p^\top)^{-1},
\end{equation}
where $\Pi_{U_f}^{\perp} = I - U_f^\top(U_fU_f^\top)^{-1}U_f$. Although the estimate $\hat{\mathcal{H}}_{fp}$ is consistent \cite{Knudsen2001consistency}, the one-step regression method cannot preserve the lower-triangular Toeplitz structure of the transmission matrix $G_f$, which is responsible for recording the impact of future input $U_f$ on future output $Y_f$. Due to the loss of this structure in $\hat G_{f}$, the model format is not causal anymore, which poses a challenge in statistical analysis. 

\begin{remark} \label{Rmk0}
In some literature of SIMs, the above one-step regression is called the projection method, in the sense that the future input $U_f$ is first projected out using
\begin{equation} \label{E12}
	Y_f\Pi_{U_f}^{\perp} = \Gamma_fL_pZ_p\Pi_{U_f}^{\perp} + H_fE_f\Pi_{U_f}^{\perp} + \Gamma_f A_K^pX_{k-p}\Pi_{U_f}^{\perp}.
\end{equation}
For a sufficiently large $p$, since $A_K^p \approx 0$, the rightmost term $\Gamma_f A_K^pX_{k-p}\Pi_{U_f}^{\perp}$ becomes negligible. Moreover, as $U_f$ is uncorrelated with $E_f$, we have $E_f\Pi_{U_f}^{\perp} \approx E_f$. After multiplying $Z_p^\top$ on both sides of \eqref{E12}, we have
\begin{equation} \label{E12a}
	Y_f\Pi_{U_f}^{\perp}Z_p^\top \approx \Gamma_fL_pZ_p\Pi_{U_f}^{\perp}Z_p^\top + H_fE_fZ_p^\top.
\end{equation}
Since $E_f$ is uncorrelated with $Z_p$, implying $\frac{1}{N}E_fZ_p^\top \approx 0$, $\Gamma_fL_p$ can then be estimated using least-squares. It is clear that the estimate of $\Gamma_fL_p$ in \eqref{E12a} is identical to \eqref{E11}.
\end{remark}

To enforce causal models, a parallel and parsimonious SIM,  PARSIM, is proposed in \cite{Qin2005novel}. Instead of using the one-step regression, PARSIM zooms into each row of \eqref{E9} and performs $f$ least-squares to estimate a bank of ARX models. To illustrate this, equation \eqref{E9} can be partitioned row-wise as 
\begin{equation} \label{E13}
	Y_{fi} = \Gamma_{fi}L_pZ_p + G_{fi}U_i + H_{fi}E_i + \Gamma_{fi} A_K^pX_{k-p},
\end{equation}
where for $i = 1,2,...f$, $\Gamma_{fi}= CA^{i-1}\in \mathbb{R}^{n_y\times n_x}$,
\begin{equation*}
	\begin{split}
		Y_{fi} &= \begin{bmatrix}
			{{y_{k+i-1}}}&{y_{k+i}}& \cdots &{y_{k+N+i-2}} \end{bmatrix}\in \mathbb{R}^{n_y\times N},\\
		U_{fi} &= \begin{bmatrix}
			{{u_{k+i-1}}}&{u_{k+i}}& \cdots &{u_{k+N+i-2}} \end{bmatrix}\in \mathbb{R}^{n_u\times N},\\
		U_i &= \begin{bmatrix}
			{{U_{f1}^{\top}}}&{{U_{f2}^{\top}}}& \cdots &{{U_{fi}^{\top}}}
		\end{bmatrix}^{\top}\in \mathbb{R}^{in_u\times N}, \\
		G_{fi} &= \begin{bmatrix}
			CA^{i-2}B&  \cdots &CB&0
		\end{bmatrix} \in \mathbb{R}^{n_y\times in_u}, \\
		H_{fi} &= \begin{bmatrix}
			CA^{i-2}K\Sigma_e^{\frac{1}{2}}&  \cdots &CK\Sigma_e^{\frac{1}{2}}&\Sigma_e^{\frac{1}{2}}
		\end{bmatrix} \in \mathbb{R}^{n_y\times in_y},
	\end{split}
\end{equation*}
where similar definitions apply to $E_{fi}$ and $E_i$. PARSIM then minimizes a bank of $i$-steps ahead prediction errors from model \eqref{E13} using ordinary least-squares (OLS),
\begin{equation} \label{E14}
	\hat \Theta_i := \begin{bmatrix}
		\widehat {\Gamma_{fi}L_p}& \hat G_{fi}
	\end{bmatrix} = Y_{fi} \begin{bmatrix}
		Z_p\\ U_i\end{bmatrix}^\dagger.
\end{equation}
At last, the whole estimate of ${\mathcal{H}}_{fp}$ is obtained by stacking the $f$ estimates together as
\begin{equation} \label{E15}
	\hat{\mathcal{H}}_{fp} = \begin{bmatrix}
		{\widehat {\Gamma_{f1}L_p}}^\top & {\widehat {\Gamma_{f2}L_p}}^\top &\cdots & {\widehat {\Gamma_{ff}L_p}}^\top \end{bmatrix}^\top.
\end{equation}
It has been shown in \cite{Qin2005novel} that the estimate in \eqref{E15} admits a smaller variance than \eqref{E11}.
\vspace{-3mm}
\subsection{Step 2: Weighted SVD} \label{Sct3.2}
Since ${\mathcal{H}}_{fp}$ has rank equal to $n_x$, to recover the extended observability matrix $\Gamma_f$ and controllability matrix $L_p$ from the estimate of ${\mathcal{H}}_{fp}$, weighted SVD is often used, i.e., 
\begin{equation} \label{E16}
	W_1\hat{\mathcal{H}}_{fp}W_2 = \hat U\hat \Lambda \hat V^\top \approx \hat U_{1}\hat \Lambda_{1}\hat V_{1}^\top, 
\end{equation}
where $\hat \Lambda_{1}$ contains the $n_x$ largest singular values. In this way, a balanced realization of $\hat \Gamma_f$ and $\hat L_p$ is
\begin{equation} \label{E17}
	\hat \Gamma_f = W_1^{-1}\hat U_{1}{\hat \Lambda}_{1}^{\frac{1}{2}}, 
	\hat L_p = {\hat \Lambda}_{1}^{\frac{1}{2}}\hat V_{1}^\top W_2^{-1}.
\end{equation}
Different choices of weighting matrices $W_1$ and $W_2$ lead to different variants of SIMs \cite{Van1995unifying,Qin2006overview}. Popular candidates of weighting matrices are summarized in Table \ref{Tab1}. \footnote{Notice that those weightings are normalized and may not be the same as they appear in the referred papers. These weightings, however, give estimates of $\Gamma_f$ and $L_p$ identical to those obtained using the original choice of weighting \cite{Jansson1998consistency}.}
\begin{table}
	\caption{Candidates of weighting matrices}
	\centering
	\begin{tabular}{ccc}
		\toprule
		Method & $W_1$ & $W_2$\\
		\midrule
		OKID \cite{Phan1995improvement} & $I$ & $I$\\
		N4SID\cite{Van1994n4sid}& $I$ & $(\frac{1}{N}Z_pZ_p^\top)^{\frac{1}{2}}$ \\
		MOESP\cite{Verhaegen1994identification}&$I$ & $(\frac{1}{N}Z_p\Pi_{U_f}^{\perp}Z_p^\top)^{\frac{1}{2}}$ \\
		IVM\cite{Viberg1995subspace}&$ (\frac{1}{N}Y_f\Pi_{U_f}^{\perp}Y_f^\top)^{-{\frac{1}{2}}}$ & $(\frac{1}{N}Z_p\Pi_{U_f}^{\perp}Z_p^\top)(\frac{1}{N}Z_pZ_p^\top)^{-\frac{1}{2}}$ \\
		CVA\cite{Larimore1990canonical}& $(\frac{1}{N}Y_f\Pi_{U_f}^{\perp}Y_f^\top)^{-{\frac{1}{2}}}$ & $(\frac{1}{N}Z_p\Pi_{U_f}^{\perp}Z_p^\top)^{\frac{1}{2}}$ \\
		\bottomrule
	\end{tabular}
 \label{Tab1}
\end{table}
\vspace{-3mm}
% 		PARSIM\cite{Qin2005novel}&$I$ & $(\frac{1}{N}Z_p\Pi_{U_f}^{\perp}Z_p^\top)^{\frac{1}{2}}$ \\
\subsection{Step 3: Realization of System Matrices} \label{Sct3.3}
Given estimates of $\Gamma_f$ and $L_p$, there are two paths to obtain the system matrices. It should be mentioned that currently there is no solid conclusion on which realization leads to a better model. To be specific, one is the CVA type which uses the following linear regressions in the output and state equations to estimate the system matrices:
\begin{subequations} \label{E19}
	\begin{align}
		Y_{f1} &= CX_k + \Sigma_e^{\frac{1}{2}}E_{f1}, \label{E19a} \\
		X_k^{+} &= AX_k+ BU_{f1} + K\Sigma_e^{\frac{1}{2}}E_{f1}, \label{E19b}
	\end{align}
\end{subequations}
where $X_k^{+}$ stacks the states for the next time instant compared to $X_k$. By replacing $X_k$ and $X_k^{+}$ with their estimates
\begin{equation} \label{E18}
	\hat X_k = \hat L_pZ_p, \hat X_k^{+} = \hat L_pZ_p^{+},
\end{equation}
where $Z_p^{+}$ is similarly defined as $X_k^{+}$, we have
\begin{subequations} \label{E20}
	\begin{align}
		\hat C &= Y_{f1} \hat X_k^\dagger, \label{E20a} \\
		\begin{bmatrix}
			\hat A& \hat B
		\end{bmatrix} &= \hat X_k^{+} \begin{bmatrix}
			\hat X_k\\U_{f1}\end{bmatrix}^\dagger. \label{E20b}
	\end{align}
\end{subequations}

Another realization method is the MOESP type, which directly extracts system matrices based on the shift invariance property of $\hat \Gamma_f$ and $\hat L_p$, i.e., 
\begin{subequations} \label{E21}
	\begin{align}
		\hat C &= \hat \Gamma_f(1:n_y,:), \label{E21a}\\
		\hat A &= (\hat \Gamma_{f}^{-})^\dagger \hat \Gamma_{f}^{+},\label{E21b}\\
		\hat B &= \hat L_{p}(:,(2p-1)n_y+1:2pn_y), \label{E21c}
	\end{align}
\end{subequations}
where $\hat \Gamma_{f}^{+}$ and $\hat \Gamma_{f}^{-}$ are the last and first $f-1$ row blocks of $\hat \Gamma_{f}$, and the indexing of matrices follows MATLAB syntax.
\begin{remark} \label{Rmk1}
	In this paper, our main interest is to estimate the system matrices $\left\{A,B,C\right\}$, and derive error bounds for them. In principle, the Kalman gain $K$ can also be obtained from the above algorithms with minor modifications. Meanwhile, there are also other methods to obtain $K$, such as solving a Riccati equation in N4SID and using QR factorization in PARSIM. To keep our results relatively compact, the estimate of $K$ and its error bound are not considered in this work.
\end{remark}

\vspace{-3mm}
\subsection{Roadmap Ahead} \label{Sct3.4}

Now we sketch the path ahead to the solution for Problem~\ref{Prom1}. In the first step that estimates the Hankel-like matrix ${\mathcal{H}}_{fp}$, we opt for PARSIM to facilitate the analysis. Parallel to the three steps in SIMs, we solve Problem~\ref{Prom1} by a three-step procedure: 

(1) Step 1: We first derive an error bound on $\hat \Theta_i$ in \eqref{E14} for every ARX model \eqref{E13}. In other words, we define the following events for $i=1,2,...,f$: 
\begin{equation} \label{E22}
	\mathcal{E}_{i,\Theta} := \left\{\norm{\hat \Theta_i - \Theta_i} \leq \epsilon_{\Theta_i}\right\},
\end{equation}
and require that $\mathbb{P}(\mathcal{E}_{i,\Theta}^c) \leq \frac{\delta}{f}$. We then utilize a norm inequality  (see Lemma \ref{LemA7}) between the block matrix $\widehat {\Gamma_{f}L_p} - {\Gamma_{f}L_p}$ and its sub-blocks $\widehat {\Gamma_{fi}L_p} - {\Gamma_{fi}L_p}$ to obtain the total bound on $\hat{\mathcal{H}}_{fp} - {\mathcal{H}}_{fp}$. This essentially requires that the intersection of $f$ events has probability $\mathbb{P}(\bigcap_{i=1}^f\mathcal{E}_{i,\Theta}) \geq 1-\delta$, which is guaranteed due to the union bound $\mathbb{P}(\bigcup_{i=1}^f\mathcal{E}_{i,\Theta}^c) \leq \delta$.

(2) Step 2: We use recent results from SVD robustness \cite{Tu2016low,Oymak2019non} to provide error bounds on the extended observability matrix $\Gamma_{f}$ and controllability matrix $L_p$, where the impact of different weighting matrices in Table \ref{Tab1} is also discussed.

(3) Step 3: We derive error bounds $\epsilon_A$, $\epsilon_B$, and $\epsilon_C$ on the system matrices $\left\{A,B,C\right\}$ coming from the Larimore and MOESP realization algorithms.

\vspace{-3mm}
\section{Finite Sample Analysis of ARX Models} \label{Sct4}
Following our roadmap, we formalize Step 1 above in this section. We emphasize that the results presented in this section apply to each ARX model in \eqref{E13} for $i=1,...,f$, where the subscript $i$ shows the model-specific dependence. For convenience, we define two covariates
\begin{subequations} \label{E23}
	\begin{align}
		\phi_{p,i}(k) &= \begin{bmatrix}y_p^\top(k)& u_p^\top(k)&u_i^\top(k)\end{bmatrix}^\top \in \mathbb{R}^{d_{p,i}}, \\
		w_{p,i}(k) &= \Lambda_{u,e}^{-1}\begin{bmatrix}
			u_p^\top(k)& u_i^\top(k)&e_p^\top(k)\end{bmatrix}^\top \in \mathbb{R}^{d_{p,i}},
	\end{align}
\end{subequations}
where $d_{p,i} = p(n_u+n_y) + in_u$ is the problem dimension of each ARX model, and $\Lambda_{e,u} = {\rm {diag}}(\sigma_u,\cdots,\sigma_u,1,\cdots,1)$ normalizes $w_{p,i}(k)$, such that it has unite variance. We further partition the following matrices column-wise:
\begin{subequations} \label{E24}
	\begin{align}
		\Phi_{p,i} = &\begin{bmatrix}\phi_{p,i}(1)&\phi_{p,i}(2)&\cdots&\phi_{p,i}(N) \end{bmatrix}, \\
		    E_i =& \begin{bmatrix}e_{i}(1)&e_{i}(2)&\cdots&e_{i}(N)\end{bmatrix}, \\
		    W_{p,i} =& \begin{bmatrix}w_{p,i}(1)& w_{p,i}(2)&\cdots&w_{p,i}(N)\end{bmatrix}.
	\end{align}
\end{subequations}
Moreover, we have the following definitions regarding the covariance and empirical covariance of $\phi_{p,i}(k)$ and $x_k$:
\begin{subequations} \label{E27}
	\begin{align}
		{\Sigma}_{p,i,k} & := \mathbb{E}\left[\phi_{p,i}(k)\phi_{p,i}^\top(k)\right],\\
		{\hat \Sigma}_{p,i,N} &:= \frac{1}{N}\sum_{k=1}^{N}\phi_{p,i}(k)\phi_{p,i}^\top(k),\\
		{\Sigma}_{x,k} &:= \mathbb{E}\left[x_{k}x_{k}^\top\right].
	\end{align}
\end{subequations}

Note that the regressor in \eqref{E14} can be rewritten as
\begin{equation}  \label{AE2}
	\Phi_{p,i}:= \begin{bmatrix}Z_p\\ U_i\end{bmatrix} = \begin{bmatrix}
	Y_p\\U_p\\ U_i\end{bmatrix}= \mathcal{O}_{p,i}X_{k-p} + \mathcal{T}_{p,i}\Lambda_{u,e}W_{p,i},	
\end{equation}
where
\begin{equation*}
	\mathcal{O}_{p,i} = \begin{bmatrix}\Gamma_p\\0\\0\\ \end{bmatrix}, W_{p,i} = \Lambda_{u,e}^{-1} \begin{bmatrix}U_p\\ U_i\\E_p\end{bmatrix}, \mathcal{T}_{p,i} = \begin{bmatrix}G_p&0&H_p \\I&0&0 \\0&I&0 \end{bmatrix}.  
\end{equation*}
Then, the error of the OLS estimate \eqref{E14} can be written as 
\begin{equation} \label{E26}
	\begin{split}
		{\tilde \Theta}_{i} := &{\hat \Theta}_{i} - \Theta_{i}= H_{fi}E_i\Phi_{p,i}^\dagger + \Gamma_{fi}A_K^pX_{k-p} \Phi_{p,i}^\dagger \\
		= &\underbrace{H_{fi}\sum_{k=1}^{N}\frac{1}{N}{e_i(k)\phi_{p,i}^\top(k)} \left({\hat \Sigma}_{p,i,N}^{-1}\right)}_{{\text{cross-term error}} \ {\tilde \Theta}_{i}^E} + \\
		&\underbrace{\Gamma_{fi}A_K^p\sum_{k=1}^{N}\frac{1}{N}{x_{k}\phi_{p,i}^\top(k)} \left({\hat \Sigma}_{p,i,N}^{-1}\right)}_{{\text{truncation bias} \ {\tilde \Theta}_{i}^B}}.
	\end{split}
\end{equation}
There are two types of errors, namely, the cross-term error ${\tilde \Theta}_{i}^E$ and the truncation bias ${\tilde \Theta}_{i}^B$. A key observation is that the future innovations $e_i(k)$ are independent of the covariate $\phi_{p,i}(l)$ for all $l<k$. This provides a martingale structure, which is convenient to analyze. To bound the above two errors, we first use results from the smallest eigenvalue of the empirical covariance of causal Gaussian processes to lower bound ${\hat \Sigma}_{p,i,N}$ \cite{Ziemann2023note,Ziemann2023tutorial,Lale2021finite}, which simultaneously establish the PE condition.
\vspace{-3mm}
\subsection{Persistence of Excitation} \label{Sct4.1}
To achieve PE, the number of samples $N$ should exceed a certain threshold, which we call the burn-in time $N_{pe}$.

\begin{definition}\label{Def1}
Given a failure probability $0<\delta<1$, a past horizon $p$, and a future horizon $i$ in each ARX model, the burn-in time $N_{pe}$ is defined as
\begin{equation} \label{E30}
	N_{pe}(\delta,p,i) = \min\left\{N: N\geq N_W\left(\delta,p,i\right), N_\Phi\left(\delta,p,i\right)\right\},
\end{equation}
where $N_W\left(\delta,p,i\right)$ and $N_\Phi\left(\delta,p,i\right)$ are defined in \eqref{AE12} and \eqref{AE19}, respectively.
\end{definition}

\begin{remark} \label{Rmk3}
	As shown in Appendix \ref{App1}, for any given $p$, $i$ and $\delta$, the $N$-dependent factors $N_W\left(\delta,p,i\right)$ and $N_\Phi\left(\delta,p,i\right)$ grow at most logarithmically with  $N$. Therefore, for a sufficiently large $N$, the existence of $N_{pe}(\delta,p,i)$ is guaranteed.
\end{remark}

A condition on PE is given the following lemma:
\begin{lemma} \label{Lem1}
	Fix a failure probability $0<\delta<1$. If $N\geq N_{pe}(\frac{\delta}{3},p,i)$, then, we have that with probability at least $1-\delta$,
	\begin{equation} \label{E31}
		{\hat \Sigma}_{p,i,N} \succcurlyeq \bar\sigma_{p,i}^2  I,
	\end{equation}
	where $\bar\sigma_{p,i} := \frac{\norm{\mathcal{T}_{p,i}}\norm{\Lambda_{u,e}}}{2}>0$.
\end{lemma}
\begin{proof}
    See Appendix \ref{App1}.
\end{proof}

\begin{remark} \label{Rmk4}
	Some relevant PE conditions appear in \cite{Tsiamis2019finite} and \cite{Lale2021finite}, where past outputs and inputs are included in regressors. Our analysis extends this by additionally incorporating future inputs, thus establishing a more general PE condition. This broader result is also useful for analyzing data-dependent weighting matrices, as detailed in Section \ref{Sct5}. Furthermore, this PE condition also holds for marginally stable systems.
\end{remark}
\vspace{-3mm}
\subsection{Bound on Cross-term Error } \label{Sct4.2}

Based on Lemma \ref{Lem1}, a upper bound on the cross-term error ${\tilde \Theta}_{i}^E$ in \eqref{E26} is provided in the following lemma:
\begin{lemma} \label{Lem2}
	Fix a failure probability $0<\delta<1$. If $N \geq N_{pe}(\frac{\delta}{9},p,i)$, then with probability at least $1-\delta$, we have
	\begin{equation} \label{E32}
		\norm{{\tilde \Theta}_{i}^E}^2 \leq c_1\frac{\epsilon_{i,E}^2}{N},
	\end{equation}
	where $\epsilon_{i,E}^2 = \frac{\norm{H_{fi}}^2}{\bar\sigma_{p,i}^2}\left(d_{p,i}{\rm{log}}\frac{d_{p,i}}{\delta} + {\rm{log}} \left({\rm{det}}\left(\frac{{\Sigma}_{p,i,N}}{\bar\sigma_{p,i}^2}\right)\right) \right)$.
\end{lemma}
\begin{proof}
	See Appendix \ref{App2}. 
\end{proof}

\vspace{-4mm}
\subsection{Bound on Truncation Bias} \label{Sct4.3}

In order to ensure that the truncation bias term ${\tilde \Theta}_{i}^B$ decays much faster than the cross-term error ${\tilde \Theta}_{i}^E$, we make the following assumption regarding the past horizon $p$.
\begin{assumption} \label{Asp2}
	The past horizon is chosen as $p=\beta\text{log}N$, where $\beta$ is large enough such that
	\begin{equation} \label{E33}
		\norm{CA_K^p}\norm{{\Sigma}_{x,N}} \leq N^{-3}. 
	\end{equation}
\end{assumption}

\begin{remark} \label{Rmk6}
	To ensure that the model \eqref{E13} closely approximates an ARX model, the truncation bias $\Gamma_{fi}A_K^px_{k}$ should be small enough, which requires that the exponentially decaying term $A_K^p$ counteracts the magnitude of the state $x_{k}$. To illustrate this, we have that 
	\begin{equation*}
		\norm{{\Sigma}_{x,N}} \leq \norm{\begin{bmatrix}B&K\end{bmatrix}\begin{bmatrix}\sigma_u^2I&0\\0&\Sigma_e\end{bmatrix}\begin{bmatrix}B^\top\\K^\top\end{bmatrix}}\sum_{k=0}^{N-1}\norm{A^k}^2. 
	\end{equation*}
	As a consequence of Lemma \ref{LemA11}, the state norm $\norm{{\Sigma}_{x,N}}$ grows at most polynomially with $N$. Meanwhile, since $\rho(A_K)<1$, we have $\norm{A_K^p} = \mathcal{O}({\bar\rho}^p)$ for some ${\bar\rho} > \rho(A_K)$. Taking $p=\beta\text{log}N$, we have $\norm{A_K^p} = \mathcal{O}(N^{-\beta/{\rm{log}(1/{\bar\rho})}})$, which implies that the condition \eqref{E33} is guaranteed for a large enough $\beta$.
\end{remark}

Under Assumption \ref{Asp2}, a bound on the bias term ${\tilde \Theta}_{i}^B$ in \eqref{E26} is provided in the following lemma:
\begin{lemma} \label{Lem3}
	Fix a failure probability $0<\delta<1$. If $N \geq N_{pe}(\frac{\delta}{9},\beta\text{log}N,i)$, then with probability at least $1-\delta$, we have
	\begin{equation} \label{E34}
		\norm{{\tilde \Theta}_{i}^B}^2 \leq c_2\frac{\epsilon_{i,B}^2}{N^2},
	\end{equation}
    where $\epsilon_{i,B}^2 = \frac{n_x}{{\bar\sigma_{p,i}^2}}{\rm{log}} \frac{1}{\delta}$.
\end{lemma}
\begin{proof}
	See Appendix \ref{App3}.
\end{proof}

\vspace{-4mm}
\subsection{Overall Bound} \label{Sct4.4}

Lemma \ref{Lem2} suggests that the cross-term error ${\tilde \Theta}_{i}^E$ decays as $\mathcal{O}(1/\sqrt{N})$, and Lemma \ref{Lem3} suggests that the truncation bias ${\tilde \Theta}_{i}^B$ decays as $\mathcal{O}(1/N)$. This implies that ${\tilde \Theta}_{i}^B$ is dominated by ${\tilde \Theta}_{i}^E$ and it can be considered negligible. After absorbing higher order terms into the dominant term by inflating the constants accordingly, we obtain the following theorem controlling the whole error ${\tilde \Theta}_{i}$ of each ARX model in our collection.
\begin{theorem} \label{The1}
	Fix a failure probability $0<\delta<1$. If $N \geq N_{pe}(\frac{\delta}{9},\beta\text{log}N,i)$, then with probability at least $1-2\delta$, we have
	\begin{equation} \label{E35}
		\norm{{\tilde \Theta}_{i}}^2 \leq c\frac{\epsilon_{i,E}^2}{N}.
	\end{equation}
\end{theorem}

After obtaining an error bound on $\tilde{\Theta}_i$ in each ARX model, we proceed to bound the total error of ${\mathcal{H}}_{fp}$, which is crucial for our subsequent analysis.  Based on the norm relation between a block matrix and its blocks in Lemma \ref{LemA7}, it is straightforward to obtain a total bound on $\hat{\mathcal{H}}_{fp} - {\mathcal{H}}_{fp}$ from each bound $\norm{{\tilde \Theta}_{i}}$.
\begin{theorem} \label{The2}
	Fix a failure probability $0<\delta<1$. If $N \geq \max\limits_{1\leq i\leq f} \left\{N_{pe}\left(\frac{\delta}{9f},\beta\text{log}N,i\right)\right\}$, then with probability at least $1-2\delta$, we have
	\begin{equation} \label{E37}
		\norm{\hat{\mathcal{H}}_{fp} - {\mathcal{H}}_{fp}} \leq \sqrt{f} \max\limits_{1\leq i\leq f} \norm{\tilde{\Theta}_i} \leq \sqrt{\frac{cf}{N}}\max\limits_{1\leq i\leq f} \epsilon_{i},	
	\end{equation}
	where $\epsilon_{i}^2 = \frac{\norm{H_{fi}}^2}{\bar\sigma_{p,i}^2} \left(d_{p,i}{\rm{log}}\frac{d_{p,i}f}{\delta} + {\rm{log}}\left({\rm{det}}\left(\frac{{\Sigma}_{p,i,N}}{\bar\sigma_{p,i}^2}\right)\right) \right)$.
\end{theorem}
\vspace{-3mm}
\section{Robustness of Balanced Realization} \label{Sct5}

Following our roadmap, having obtained an overall error bound on $\hat{\mathcal{H}}_{fp} - {\mathcal{H}}_{fp}$ in Step 1, we now move to Step 2 to derive error bounds on the extended controllability and observability matrices, and Step 3 to obtain error bounds on the system matrices.
\vspace{-3mm}
\subsection{Weighted Singular Value Decomposition} \label{Sct5.1}

Weighted SVD is crucial for improving the performance of SIMs. As summarized in Table \ref{Tab1}, different choices of weighting matrices lead to different variants of SIMs. Since those data-dependent weighting matrices share a similar structure, we choose the pair used in MOESP and PARSIM to illustrate their characteristics, where
\begin{equation} \label{E38}
	W_1 = I, W_2 = (\frac{1}{N}Z_p\Pi_{U_f}^{\perp}Z_p^\top)^{\frac{1}{2}}.
\end{equation}
The focus is on the data-dependent $W_2$, whose finite-sample properties are summarized as follows: \footnote{Similarly, conditions for other weighting matrices in Table \ref{Tab1} can be obtained.} 
\begin{lemma} \label{Lem4}
	Fix a failure probability $0<\delta<1$. If $N \geq N_{pe}(\delta,p,f)$, then with probability at least $1-3\delta/2-\delta_u$, where $\delta_u = (2(N+f-1)n_u)^{-{\rm{log}}^2(2fn_u){\rm{log}}(2(N+f-1)n_u)}$, the weighting matrix $W_2$ in \eqref{E38} satisfies:
	
	(1) $W_2$ is positive definite.

	(2) $\norm{W_2}$ grows at most logarithmically with $N$.

	(3) $\norm{W_2^{-1}}$ is bounded.
\end{lemma}
\begin{proof}
	See Appendix \ref{App4}.
\end{proof}

Now we study the robustness of SVD. Since other weighting matrices in Table \ref{Tab1} have the same properties as in Lemma \ref{Lem4}, we will henceforth not specify a pair but use $W_1$ and $W_2$ to represent them universally. For simplicity, we further assume that $W_1$ and $W_2$ satisfy the conditions in Lemma \ref{Lem4} almost surely. The weighted SVD of the true value of ${\mathcal{H}}_{fp}$ is 
\begin{equation} \label{E39}
	W_1{\mathcal{H}}_{fp}W_2 = \begin{bmatrix}
		U_{1}&U_0
	\end{bmatrix}\begin{bmatrix}
		\Lambda_{1}&0\\0&0
	\end{bmatrix}\begin{bmatrix}
		V_{1}&V_0
	\end{bmatrix}^\top,
\end{equation}
where $\Lambda_{1} \succ 0$ contains the $n_x$ largest singular values. A balanced realization for $\Gamma_{f}$ and $L_p$ is 
\begin{equation} \label{E40}
	\bar{\Gamma}_{f} = W_1^{-1}U_{1}\Lambda_{1}^{1/2}, \bar{L}_p = \Lambda_{1}^{1/2}V_{1}^\top W_2^{-1}.		
\end{equation}
Then, we obtain the following robustness results regarding the estimates of  ${\Gamma}_{f}$ and $L_p$.
\begin{theorem} \label{The3}
	If the following condition is satisfied:
	\begin{equation} \label{E41}
		\norm{W_1\hat{\mathcal{H}}_{fp}W_2- W_1{\mathcal{H}}_{fp}W_2} \leq \frac{\sigma_{n_x} (W_1{\mathcal{H}}_{fp}W_2)}{4},
	\end{equation}
	then there exists an orthogonal matrix $T$, such that for a failure probability $0< \delta < 1$, if $N \geq \max \limits_{1\leq i\leq f} \left\{N_{pe}(\frac{\delta}{9f},\beta {\rm{log}} N,i)\right\}$, then with probability at least $1-2\delta$, we have
	\begin{subequations} \label{E42}
		\begin{align}
			\norm{\hat{\Gamma}_{f}- \bar{\Gamma}_{f}T} &\leq \sqrt{\frac{40n_x}{\sigma_{n_x} ({\mathcal{H}}_{fp})}} \norm{\hat{\mathcal{H}}_{fp}- {\mathcal{H}}_{fp}}W_{\Gamma},\\
			\norm{\hat{L}_p- T^\top\bar{L}_p} &\leq \sqrt{\frac{40n_x}{\sigma_{n_x} ({\mathcal{H}}_{fp})}} \norm{\hat{\mathcal{H}}_{fp}- {\mathcal{H}}_{fp}}W_{L},
		\end{align}
	\end{subequations}
	where $W_{\Gamma} = \norm{W_1}\norm{W_2}\norm{W_1^{-1}}^\frac{3}{2}\norm{W_2^{-1}}^\frac{1}{2}$ and $W_{L} = \norm{W_1}\norm{W_2}\norm{W_2^{-1}}^\frac{3}{2}\norm{W_1^{-1}}^\frac{1}{2}$.
\end{theorem}
\begin{proof}
	See Appendix \ref{App4}.
\end{proof}
\begin{remark} \label{Rmk10}
Compared to a stochastic non-singular matrix $T$ in Problem \ref{Prom1}, matrix $T$ is constrained to be an orthogonal matrix in Theorem \ref{The3}. This is mainly due to Lemma \ref{LemA9}. Furthermore, according to the eigenvalue decomposition, every non-singular matrix has an associated orthogonal matrix. Therefore, such a constraint will not affect the generality of our results.
\end{remark}
\vspace{-3mm}
\subsection{Bounds on System Matrices} \label{Sct5.2}
Having obtained upper bounds on the extended observability matrix $\bar{\Gamma}_{f}$ and controllability matrix $\bar{L}_p$ in Step 2, we now move to the final Step 3 to derive error bounds on the system matrices, where two realization algorithms are studied.

\subsubsection{CVA Type Realization}
The error bounds on system matrices from the realization algorithm \eqref{E20} are follows:
\begin{theorem} \label{The4}
	If the condition \eqref{E41} is satisfied, then there exists an orthogonal matrix $T$, such that for a failure probability $0< \delta < 1$, if $N \geq \max \limits_{1\leq i\leq f} \left\{N_{pe}(\frac{\delta}{9f},\beta {\rm{log}} N,i)\right\}$, then with probability at least $1-2\delta$, we have
	\begin{subequations} \label{E43}
		\begin{align}
			&\norm{\hat{C}- \bar{C}T} \leq c_4 \norm{\hat{L}_p- T^\top \bar L_p}  + c_5 \frac{\epsilon_{0,B}}{N} + c_6 \frac{\epsilon_{0,E}}{\sqrt{N}}, \label{E43a}\\
			\nonumber
			&{\rm {max}}\left\{\norm{\hat{A}- T^\top \bar{A}T},\norm{\hat{B}- T^\top\bar{B}} \right\} \leq \\
			& c_7 \norm{\hat{L}_p- T^\top \bar L_p}  + c_8 \frac{\epsilon_{1,B}}{N} + c_9 \frac{\epsilon_{1,E}}{\sqrt{N}}, \label{E43b}
		\end{align}
	\end{subequations}
	where detailed expressions of $c_4,\dots,c_9$, $\epsilon_{0,B}$, $\epsilon_{0,E}$, $\epsilon_{1,B}$ and $\epsilon_{1,E}$ are given in Appendix \ref{App5}.
\end{theorem}
\begin{proof}
	See Appendix \ref{App5}.
\end{proof}

\subsubsection{MOESP Type Realization}
The error bounds on system matrices from the realization algorithm \eqref{E21} are follows:
\begin{theorem} \label{The5}
	If the condition \eqref{E41} is satisfied, then there exists an orthogonal matrix $T$, such that for a failure probability $0< \delta < 1$, if $N \geq \max \limits_{1\leq i\leq f} \left\{N_{pe}(\delta/(9f),\beta {\rm{log}} N,i)\right\}$, then with probability at least $1-2\delta$, we have
	\begin{subequations} \label{E44}
		\begin{align}
			\norm{\hat{C}- \bar{C}T} &\leq \norm{\hat{\Gamma}_{f}- \bar{\Gamma}_{f}T}, \label{E44a} \\
			\norm{\hat{B}- T^\top\bar{B}} &\leq \norm{\hat{L}_p- T^\top\bar{L}_p}, \label{E44b} \\
			\norm{\hat{A}- T^\top \bar{A}T} &\leq \frac{\sqrt{\norm{\Gamma_{f}L_p}}+\sigma_o}{\sigma_o^2}\norm{\hat{\Gamma}_{f}- \bar{\Gamma}_{f}T}, \label{E44c}
		\end{align}
	\end{subequations}
	where $\sigma_o = {\rm{min}}\left(\sigma_{n_x}(\hat{\Gamma}_f^{-}),\sigma_{n_x}(\bar{\Gamma}_f^{-})\right)$.
\end{theorem}
\begin{proof}
	See Appendix \ref{App5}.
\end{proof}
\vspace{-3mm}
\section{Discussions} \label{Sct6}

We now discuss the implications of our main results. Specifically, we will answer two key questions: how to extend our method to other variants of SIMs, and what does the finite sample analysis bring. 
\vspace{-3mm}
\subsection{Extensions to Other Methods} \label{Sct6.1}
It is clear that our results in Steps 2 and 3 cover a large class of SIMs. In Step 1, to strictly enforce a causal model, we opt for the multi-regression method used in PARSIM--one of the most appealing SIMs, to estimate the Hankel-like matrix ${\mathcal{H}}_{fp}$. The following observations suggest that the technique used in the analysis of PARSIM can be extended to other SIMs.

First, our analysis for PARSIM can be extended to the one-step regression method in \eqref{E11}. If we study the estimation error of ${\mathcal{H}}_{fp}$ in \eqref{E11} separately, the cross-term error will be  $E_f\Pi_{U_f}^{\perp}Z_p^\top(Z_p\Pi_{U_f}^{\perp}Z_p^\top)^{-1}$.  Due to the data-dependent projection matrix $\Pi_{U_f}^{\perp}$, the columns of $E_f$ and $Z_p$ are mixed together, bringing challenges to statistical analysis. However, this problem can be avoided if we study the total error of $\Theta$ in \eqref{E10}, which is equivalent to setting $i=f$ in PARSIM.

Second, our analysis can be extended to other SIMs that estimate high order ARX models using least-squares in their first step, such as SSARX\cite{Jansson2003subspace} and PBSID \cite{Chiuso2005consistency}. To be specific, for SSARX, it first estimates the predictor Markov parameters $\left\{{C{A_K^{i}}B,C{A_K^{i}}K}\right\}_{i=0}^{f-1}$ from a high-order ARX model, and then replaces the true Markov parameters in the transmission matrices with their estimates, and proceeds to estimate a Hankel-like matrix. PBSID, also known as the whitening filter approach \cite{Chiuso2005consistency}, starts from the predictor form \eqref{E1A}. Similar to PARSIM, it utilizes the structure of the transmission matrices to carry out multiple regressions parallelly. It is clear that our methods can be applied to the first step of SSARX, and to every step of PBSID in the ope-loop setting.
\vspace{-3mm}
\subsection{What Does Finite Sample Analysis Bring}  \label{Sct6.2}
Under the umbrella of this question, we analyze the implications of our results and validate them with simulations on a benchmark SISO system used in SIM studies \cite{Chiuso2007role}. It should be mentioned that the results in this paper directly extend to MIMO systems. The SISO system is given by
\begin{equation} \label{E45}
	y_k + ay_{k-1} = bu_{k-1} + e_k + ce_{k-1},
\end{equation}
where $a = -0.7$, $b = 1$ and $c = 0.5$. The innovations $e_{k} \sim \mathcal{N}(0,4)$. Two types of inputs are considered, one is a white input given by $u_k \sim \mathcal{N}(0,1)$, and the other is a colored input\footnote{It is important to note that due to Assumption \ref{Asp1}, our theoretical results do not apply to scenarios with colored inputs yet. However, using colored inputs in our simulations helps in demonstrating the behavior of SIMs.}, generated by a white noise $r_{k}\sim \mathcal{N}(0,1)$ passing through a filter $H_u(q^{-1}) =  \frac{0.318}{1-0.5q^{-1}+0.9q^{-2}}$, where $q^{-1}$ is the backward shift operator.
\subsubsection{Statistical Rates}
According to Theorems \ref{The2}, \ref{The3}, \ref{The4} and \ref{The5}, the error bounds on the Hankel-like matrix ${\mathcal{H}}_{fp}$, the extended observability matrix $\Gamma_{f}$, controllability matrix ${L}_p$, and system matrices $C$, $B$ and $A$ decay as $\mathcal{O}(1/\sqrt{N})$ up to logarithmic terms. Classical asymptotic results, such as the LIL, can tighten the log factor to $\mathcal{O}(\text{log}\text{log}N/\sqrt{N})$ \cite{Deistler1995consistency}. This suggests that our bounds are not tight, which is one of the downsides of the non-asymptotic analysis. It is possible to optimize our results in the future.

\subsubsection{Persistence of Excitation}
Similar to the PE condition for consistency analysis in the asymptotic regime \cite{Jansson1998consistency}, Lemma~\ref{Lem1} provides a non-asymptotic PE condition. Specifically, it guarantees the invertibility of the empirical covariance matrix ${\hat \Sigma}_{p,i,N}$ defined in \eqref{E27} by establishing a lower bound on its smallest eigenvalue. In practice, once data is available, one can directly verify PE by computing the smallest eigenvalue of ${\hat \Sigma}_{p,i,N}$. Nonetheless, Lemma \ref{Lem1} provides a theoretical threshold, denoted by the burn-in time $N_{pe}$ in \eqref{E30}, indicating the minimum sample size required to ensure invertibility of ${\hat \Sigma}_{p,i,N}$. Such a non-asymptotic result is valuable for designing experiments (e.g., determining excitation duration), and implementing stopping rules in adaptive control \cite{Jedra2022finite}.

\subsubsection{Dimensional Dependence}
As shown in \eqref{E30} and Theorem~\ref{The2}, the burn-in time $N_{pe}$ and error bounds scale with the problem dimension $d_{p,i}$ and the state dimension $n_x$. Such a dimensional dependence still holds when $n_x$ increases to the same order as $N$, whereas the results in the asymptotic regime are less meaningful in this case \cite{Tsiamis2023statistical}.

\subsubsection{A Sweet Spot for the Past Horizon $p$}
Assumption~\ref{Asp2} implies that to make the truncation bias ${\tilde \Theta}_{i}^B$ decay much faster than the cross-term error ${\tilde \Theta}_{i}^E$, the past horizon $p$ should increase at a proper rate with $N$, i.e., $p=\beta\text{log}N$, where $\beta$ is sufficiently large. Meanwhile, a larger $p$ means that there are more parameters to be estimated, thus implying a larger error bound. This highlights that, for a fixed $N$, there is a sweet spot for the choice of $p$. Similar conclusions are found in asymptotic analysis \cite{Chiuso2007role,Galrinho2018parametric}, suggesting that $p$ should grow moderately with $N$, neither too slow nor too fast. In practice though, $p$ can be selected using information criteria or a two-step procedure--either by minimizing the prediction error of the estimated state-space model or by directly minimizing the error bounds. We use the numerical example \eqref{E45} to demonstrate the existence of the sweet spot. We fix the future horizon at $f=7$, and vary the number of samples $\bar N = 500:1000:2500$ and the past horizon $p = 2:4:30$. The input is white, and the weighting matrices are chosen as $W_1=I$ and $W_2=I$. We run 100 Monte Carlo trials. The performance is evaluated by the average normalized error of the poles $\norm{\hat{a} - a}/\norm{a}$, where $\hat{a}$ is obtained using two realization algorithms. As shown in Figure \ref{F4}, for both two realization methods, when the number of samples is fixed, there is a sweet spot for $p$ that minimizes the errors. In addition, when the number of samples increases, the sweet spot for $p$ tends to increase as well.
\begin{figure}
	\centering
	\includegraphics[scale=0.5]{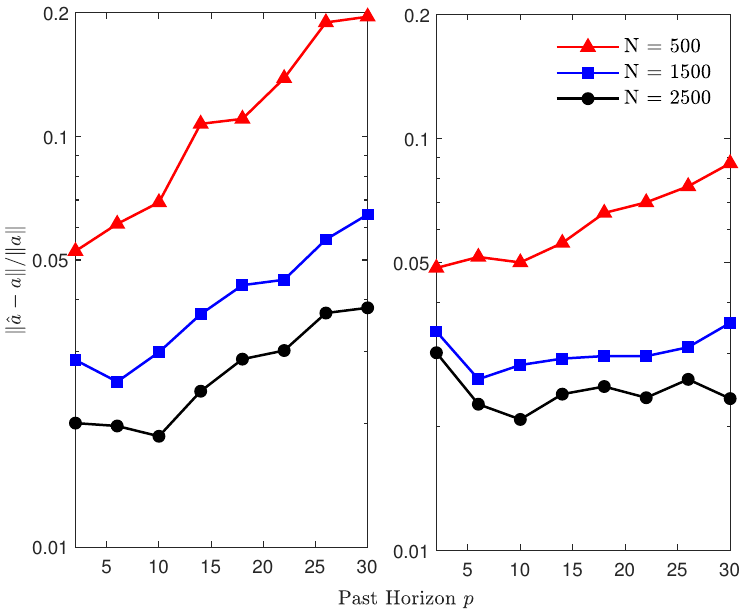}
	\caption{A sweet spot for past horizon: MOESP (left) and CVA (right) realizations.}
	\label{F4}
\end{figure}

\subsubsection{On the Impact of Weighting Matrices}
In the asymptotic regime, the impact of weighting matrices is discussed in \cite{Bauer2000impact,Bauer2002some,Bauer2005asymptotic,Gustafsson2002subspace}. At a high level, $W_1$ is related to a maximum likelihood or CVA objective, while $W_2$ is related to an orthogonal projection. In addition, $W_1$ has little impact on the asymptotic accuracy of $\Gamma_{f}$, and $W_2$ has little impact on the asymptotic accuracy of $L_{p}$. As shown in Theorem \ref{The3}, our work provides a new perspective on the impact of weighting matrices. To be specific, for the robustness of SVD, the condition \eqref{E41} should be satisfied, which guarantees that the singular vectors related to small singular values of $W_1{\mathcal{H}}_{fp}W_2$ can be separated from the singular vectors coming from the noise $W_1\left(\hat{\mathcal{H}}_{fp}- {\mathcal{H}}_{fp}\right)W_2$. Because $W_1$ and $W_2$ reshape these singular values, they can either relax or tighten the condition \eqref{E41}.
To see this clearly, we use the numerical example \eqref{E45} to show the impact of different weighting matrices in Table~\ref{Tab1}. The simulation settings are the same as before, and the past horizon is fixed at $p=f=7$. Both white input and colored input are considered. The performance is evaluated by the ratio 
\begin{equation}
	\kappa := \frac{\norm{W_1(\hat{\mathcal{H}}_{fp}-{\mathcal{H}}_{fp})W_2}}{\sigma_{n_x} (W_1{\mathcal{H}}_{fp}W_2)}.
\end{equation}
We determine whether the condition \eqref{E41} is satisfied or not by comparing $\kappa$ with $1/4$. The results are shown in Figure \ref{F2}\footnote{Note that N4SID gives almost identical results as MOESP, and IVM gives almost identical results as CVA, so only results for MOESP, CVA and OKID are presented.}. 
\begin{figure}
	\centering
	\includegraphics[scale=0.5]{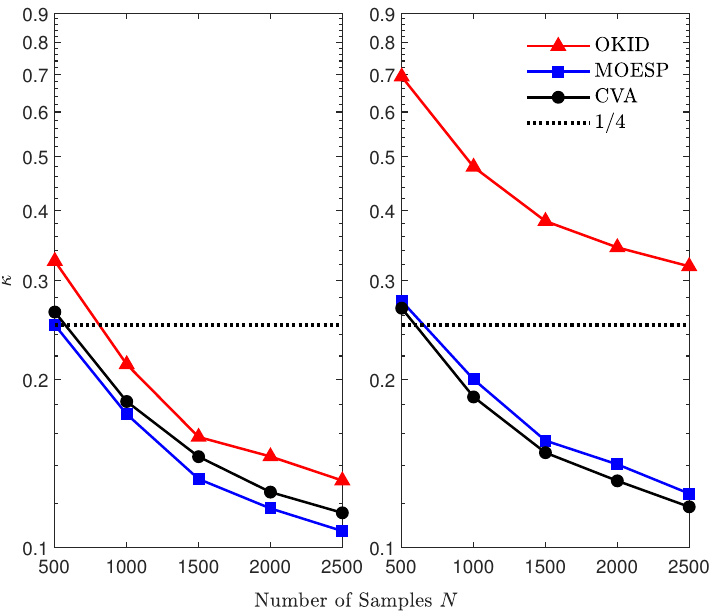}
	\caption{Robustness condition: white input (left) and colored input (right), where the three curves are from three types of weighting matrices.}
	\label{F2}
\end{figure}

First, Figure \ref{F2} suggests that different weighting matrices result in different number of samples required for the condition \eqref{E41} to be satisfied. Since $\norm{\hat{\mathcal{H}}_{fp}-{\mathcal{H}}_{fp}}$ decays as $\mathcal{O}(1/\sqrt{N})$, no matter what pair of weighting matrices we choose, $\kappa \leq 1/4$ will be eventually satisfied as $N \to \infty$. However, a good choice of weighting matrices makes it easier to satisfy this condition, such as the MOESP weighting.

Second, both weighting matrices $W_1$ and $W_2$ affect the robustness condition. As shown in Figure \ref{F2}, MOESP and CVA employ different $W_1$ and the same $W_2$, which result in different robustness conditions. Meanwhile, MOESP and OKID employ different $W_2$ and the same $W_1$, which also result in different robustness conditions. 

Third, the impact of weighting matrices is input-dependent. As shown in Figure \ref{F2}, it is easier for the MOESP weighting to satisfy the condition \eqref{E41} when the input is white than colored.

Fourth, besides their impact on the robustness condition, the weighting matrices also influence the estimation accuracy. It should be emphasized that a pair of weighting matrices making the robustness condition easier to achieve does not mean that they also imply a smaller estimation error. To illustrate this, we choose the estimate of poles coming from two realization algorithms to demonstrate the impact of the weighting matrices. The simulation settings are same as before, and only the white input is considered. The performance is evaluated by the normalized error of the poles $\norm{\hat{a} - a}/\norm{a}$. The results are shown in Figure \ref{F3}. Based on the right subplots of Figures \ref{F2} and \ref{F3}, we see that compared to the CVA weighting, although the MOESP weighting makes the robustness condition easier to achieve, it increases the estimation error of the poles.
\begin{figure}
	\centering
	\includegraphics[scale=0.5]{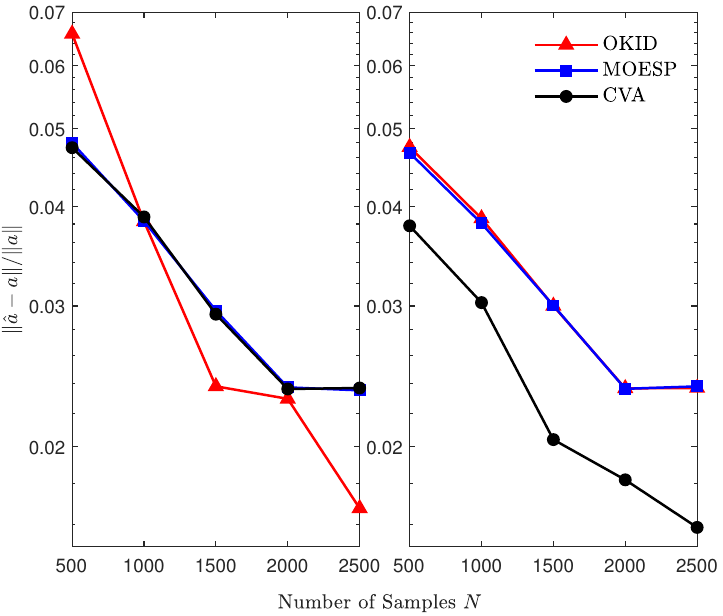}
	\caption{Normalized error of poles: MOESP (left) and Larimore (right), where the three curves are from three types of weighting matrices\protect\footnotemark.}
	\label{F3}
\end{figure}
\footnotetext{Although the OKID weighting performs better for the MOESP type of realization when $N$ becomes larger in this simple example, it generally does not outperform other methods in most cases. Additionally, since we estimate ${\mathcal{H}_{f,p}}$ using PARSIM, the best results of OKID is primarily attributable to PARSIM rather than the original OKID approach which estimates ${\mathcal{H}_{f,p}}$ in a different way \cite{Phan1995improvement}.}

In addition, as shown in Figure \ref{F3}, $W_1$ has minor influence on the estimate of the poles for the MOESP type realization, and $W_2$ has minor influence on the estimate of the poles for the Larimore type realization. This is consistent with an analysis in the asymptotic regime \cite{Bauer2000impact,Bauer2002some,Bauer2005asymptotic,Gustafsson2002subspace}. However, this conclusion cannot be obtained through our finite sample analysis. This comparison underscores the view that both asymptotic and non-asymptotic methods are valuable in uncovering the statistical properties of SIMs, and they complement each other.

Finally, we remark that the robustness condition \eqref{E41} is a sufficient condition, and our results are upper bounds. It is not sufficient to determine the best choice of weighting matrices solely based on the criteria of facilitating the achievement of robustness conditions and minimizing the upper bounds. To fully grasp the influence of the weighting matrices and develop an optimal choice, further study is needed.
\vspace{-3mm}
\section{Conclusion} \label{Sct7}

This paper presents a finite sample analysis for a large class of open-loop SIMs. Compared with the-state-of-art that mainly analyzes the performance
of the Ho-Kalman algorithm or similar variants, we investigate one of the most representative SIMs, PARSIM. Our analysis establishes a more general PE condition, and takes the different weighting matrices and two realization algorithms into account. It not only confirms that the convergence rates for estimating the Markov parameters and system matrices are $\mathcal{O}(1/\sqrt{N})$ even in the presence of inputs, in line with classical asymptotic results, but it also provides high-probability upper bounds for these estimates. Our findings complement the existing asymptotic results, and methodologies can be similarly applied to many variants of SIMS, such as classical SIMs, SSARX and PBSID.
\vspace{-3mm}
\section*{References}
\bibliographystyle{IEEEtran}
\bibliography{refs}

\appendices 
\numberwithin{equation}{section}
\vspace{-3mm}
\section{Burn-in Time and Persistence of Excitation} \label{App1}
\subsection{Proof of Lemma \ref{Lem1}}

Based on \eqref{AE2}, we have that 
\begin{equation}  \label{AE3}
	\begin{split}
		N{\hat \Sigma}_{p,i,N} = &\mathcal{O}_{p,i}X_{k-p}X_{k-p}^\top\mathcal{O}_{p,i}^\top + \mathcal{T}_{p,i}\Lambda_{u,e}W_{p,i}W_{p,i}^\top\Lambda_{u,e}^\top\mathcal{T}_{p,i}^\top + \\ 
		&\mathcal{O}_{p,i}X_{k-p}W_{p,i}^\top\Lambda_{u,e}^\top\mathcal{T}_{p,i}^\top +\mathcal{T}_{p,i}\Lambda_{u,e}W_{p,i} X_{k-p}^\top\mathcal{O}_{p,i}^\top.
	\end{split}    	
\end{equation}
In the following, we use a three-step procedure to show that ${\hat \Sigma}_{p,i,N}$ is invertible with high probability. 

Step 1 (Noises and inputs PE): Here we establish the PE condition for the corrections of noises and inputs, i.e., the term $\mathcal{T}_{p,i}\Lambda_{u,e}W_{p,i}W_{p,i}^\top\Lambda_{u,e}^\top\mathcal{T}_{p,i}^\top$. Since $\mathcal{T}_{p,i}\Lambda_{u,e}$ is a full-row rank matrix, we mainly show that the following event occurs with high probability:
\begin{equation} \label{AE4}
	\mathcal{E}_{i,W} = \left\{W_{p,i}W_{p,i}^\top = \sum_{k=1}^{N}w_{p,i}(k)w_{p,i}^\top(k) \succcurlyeq \frac{N}{2}I \right\}.
\end{equation}
Since $W_{p,i}$ contains two block Hankel matrices, existing lower bounds on individual Hankel matrices in \cite{Tsiamis2019finite,Oymak2021revisiting} cannot be directly used. We propose another approach to bound it. Based on \eqref{E23}, we have $\mathbb{E}\left[w_{p,i}(k)w_{p,i}^\top(k)\right] = I$, which further gives
\begin{equation} \label{AE5}
	\sigma_{\min}\left(\mathbb{E}\left[w_{p,i}(k)w_{p,i}^\top(k)\right]\right) = 1.
\end{equation}    
We now provide a upper bound for $\norm{w_{p,i}(k)}$. Since $e_{k} \sim \mathcal{N}(0,I)$ and $u_{k} \sim \mathcal{N}(0,\sigma_u^2I)$, we conclude that $u_k$ and $e_k$ are component-wise sub-Gaussian random variables with variance $\sigma_u^2$ and $1$, respectively. Therefore, for all $1\leq k \leq N$, with probability at least $1-\delta$, we have 
\begin{subequations} \label{AE6}
	\begin{align}
		\sigma_u^{-1}\norm{u_k} &\leq \bar u := n_u\sqrt{2n_u{\rm {log}}(\frac{32n_uN}{\delta})},  \\
		\norm{e_k} &\leq \bar e := n_y\sqrt{2n_y{\rm {log}}(\frac{32n_yN}{\delta})}.
	\end{align}
\end{subequations}
Furthermore, it is straightforward to see that 
\begin{equation}\label{AE7}
	\norm{w_{p,i}(k)} \leq \bar w_{p,i} := \bar e \sqrt{p} + \bar u \sqrt{p+i}.
\end{equation}
We now define 
\begin{subequations} \label{AE8}
	\begin{align}
		r_{p,i}(k) &= w_{p,i}(k)w_{p,i}^\top(k) - \mathbb{E}\left[w_{p,i}(k) w_{p,i}^\top(k)\right],  \\
		R_{p,i} &= \sum_{k=1}^{N}r_{p,i}(k),
	\end{align}
\end{subequations}
and their truncation
\begin{subequations} \label{AE8a}
	\begin{align}
		\tilde{r}_{p,i}(k) &= {r}_{p,i}(k)\mathbb{I}_{\left\{r_{p,i}(k) \preccurlyeq  2 \bar w_{p,i}^2I\right\}}, \\
		\tilde{R}_{p,i} &= \sum_{k=1}^{N}\tilde{r}_{p,i}(k).		
	\end{align}
\end{subequations}
Further, we define a upper bound $\bar r_{p,i} = 2\bar w_{p,i}^2 \sqrt{2{\rm{log}}\left(\frac{d_i}{\delta}\right)}$. According to Lemma \ref{LemA2-1}, we have
\begin{equation} \label{AE9}
	\begin{split}
		\mathbb{P}\left(R_{p,i} \succ \bar r_{p,i}\sqrt{N}I\right) \leq &\mathbb{P}\left(\tilde{R}_{p,i} \succ \bar r_{p,i}\sqrt{N}I\right) + \\
		&\mathbb{P}\left(\max\limits_{1\leq k\leq N} r_{p,i}(k) \succcurlyeq 2 \bar w_{p,i}^2I\right). 
	\end{split}	   	
\end{equation}
Based on Lemma \ref{LemA2} and \eqref{AE7}, each term on the right hand side of \eqref{AE9} is bounded by $\delta$. Thus, with probability of at least $1-2\delta$, we have 
\begin{equation} \label{AE10}
	\lambda_{\rm{max}}\left(R_{p,i}\right) \leq \bar r_{p,i}\sqrt{N}.
\end{equation}
Combining \eqref{AE5} and \eqref{AE10}, and using Weyl's inequality in Lemma \ref{LemA3}, we have that with probability at least $1-2\delta$,
\begin{equation} \label{AE11}
	\sigma_{\rm{min}}\left(\sum_{k=1}^{N}w_{p,i}(k)w_{p,i}^\top(k)\right) \geq N - \bar r_{p,i}\sqrt{N}.
\end{equation}
Define 
\begin{equation} \label{AE12}
	N_W\left(\delta,p,i\right) := 4\bar r_{p,i}^2.
\end{equation}
For $N\geq N_W\left(\delta,p,i\right)$, we then have that with probability at least $1-2\delta$, 
\begin{equation} \label{AE13}
	\sigma_{\rm{min}}\left(\sum_{k=1}^{N}w_{p,i}(k)w_{p,i}^\top(k)\right) \geq \frac{N}{2} > 0.
\end{equation}

Step 2 (Cross terms are small): We now show that the cross term $\mathcal{O}_{p,i}X_{k-p}W_{p,i}^\top\Lambda_{u,e}^\top\mathcal{T}_{p,i}^\top$ is much smaller and its norm increases with a rate of at most $\mathcal{O}(\sqrt{N})$ up to logarithmic terms. We first use the Markov inequality in Lemma \ref{LemA4} to upper bound state correlations $X_{k-p}X_{k-p}^\top$, which states that the following event holds with probability $1-\frac{\delta}{2}$:
\begin{equation} \label{AE14}
	\mathcal{E}_{i,X} := \left\{X_{k-p}X_{k-p}^\top  \preccurlyeq N\frac{2n_x}{\delta} {\Sigma}_{x,N}\right\}.
\end{equation}
Using Lemma \ref{LemA5} we have that with probability $1-\frac{\delta}{2}$:
\begin{equation} \label{AE15}
	\begin{split}
		&\norm{\left(X_{k-p}X_{k-p}^\top+NI\right)^{-\frac{1}{2}}X_{k-p}W_{p,i}^\top}^2  \leq \\
		&4{\text{log}\frac{\text{det}\left(X_{k-p}X_{k-p}^\top+NI\right)}{\text{det}\left(NI\right)}} + 8\text{log}\frac{5^{d_i}2}{\delta}.
	\end{split}	
\end{equation}
Furthermore, conditioned on $\mathcal{E}_{i,X}$, inequality in \eqref{AE15} can be relaxed to that with probability $1-\delta$:
\begin{equation} \label{AE16}
	\mathcal{E}_{i,XE} := \left\{\norm{\left(X_{k-p}X_{k-p}^\top+NI\right)^{-\frac{1}{2}}X_{k-p}W_{p,i}^\top}^2 \leq C_{i,XE}\right\},
\end{equation}
where $C_{i,XE} = 4{\text{log}\left(\text{det}(I + \frac{2n_x}{\delta} \Sigma_{x,N})\right)} + 8\text{log}\frac{5^{d_i}2}{\delta}$. In this way, we have that with probability $1-\delta$,
\begin{equation} \label{AE17}
	\norm{X_{k-p}W_{p,i}^\top} \leq \norm{\left(X_{k-p}X_{k-p}^\top+NI\right)^{\frac{1}{2}}}\sqrt{C_{i,XE}},  	
\end{equation}
which implies that $\mathcal{O}_{p,i}X_{k-p}W_{p,i}^\top\Lambda_{u,e}^\top\mathcal{T}_{p,i}^\top$ increases at most $\mathcal{O}(\sqrt{N})$ up to logarithmic terms.

Step 3 (Inputs and outputs PE): From Steps 1 and 2, we summarize that for $N\geq N_W\left(N,\delta,p,i\right)$, the event $\mathcal{E}_{i,W} \cap \mathcal{E}_{i,XE}$ occurs with probability at least $1-3\delta$. In this step, we show that if the number of samples is large enough, then,
\begin{equation*}
	N{\hat \Sigma}_{p,i,N}  \succcurlyeq \mathcal{O}_{p,i}X_{k-p}X_{k-p}^\top\mathcal{O}_{p,i}^\top + \mathcal{T}_{p,i}\Lambda_{u,e}W_{p,i}W_{p,i}^\top\Lambda_{u,e}^\top\mathcal{T}_{p,i}^\top,   	
\end{equation*}
holds with high probability. For an arbitrary unit vector $v \in \mathbb{R}^{d_i}$, we define
\begin{subequations} \label{AE18}
	\begin{align}
		\alpha_i &:= \frac{1}{N}v^\top\mathcal{O}_{p,i}X_{k-p}X_{k-p}^\top\mathcal{O}_{p,i}^\top v, \\
		\beta_i &:=  \frac{1}{N}v^\top\mathcal{T}_{p,i}\Lambda_{u,e}W_{p,i}W_{p,i}^\top\Lambda_{u,e}^\top\mathcal{T}_{p,i}^\top v, \\
		\gamma_{i,N} &:= 2\norm{\mathcal{T}_{p,i}}\norm{\Lambda_{u,e}}\sqrt{\frac{C_{i,XE}}{N}}.
	\end{align}    	
\end{subequations}
Then, from \eqref{AE3}, we have that
\begin{equation*}
	\begin{split}
		v^\top{\hat \Sigma}_{p,i,N} v & = \alpha_i + \beta_i + 2v^\top\mathcal{O}_pX_{k-p}W_{p,i}^\top\Lambda_{u,e}^\top\mathcal{T}_{p,i}^\top v \\
		&\geq \alpha_i + \beta_i - \gamma_{i,N}\sqrt{1+\alpha_i},
	\end{split}    	
\end{equation*}
where the second inequality is due to \eqref{AE17}. Moreover, the event $\mathcal{E}_{i,W}$ implies that $\beta_i \geq \frac{\norm{\mathcal{T}_{p,i}}^2\norm{\Lambda_{u,e}}^2}{2} > 0$. Let $N_\Phi\left(\delta,p,i\right)$ be such that 
\begin{equation} \label{AE19}
	N_\Phi\left(\delta,p,i\right) := \min \left\{N: \gamma_{i,N} \leq \min\left\{1,\bar\sigma_{p,i}^2\right\}\right\}.
\end{equation}
Since $C_{i,XE}$ grows at most logarithmically with $N$, $N_\Phi\left(\delta,p,i\right)$ always exist. Using Lemma \ref{LemA12}, we conclude that with probability at least $1-3\delta$:
\begin{equation} \label{AE20}
	\alpha_i + \beta_i - \gamma_{i,N}\sqrt{1+\alpha_i} \geq \frac{\alpha_i + \beta_i}{2} \geq \frac{\beta_i}{2},
\end{equation}
where $\frac{\beta_i}{2}  = \frac{\norm{\mathcal{T}_{p,i}}^2\norm{\Lambda_{u,e}}^2}{4} := \bar\sigma_{p,i}^2$.
\hfill $\blacksquare$
\vspace{-3mm}
\section{Bound on Cross-term Error}  \label{App2}
\subsection{Proof of Lemma \ref{Lem2}}

To bound the cross-term error ${\tilde \Theta}_{i}^E$, we first define the following three events, where each of them occurs with probability at leat $1-\delta/3$:
\begin{equation}
	\nonumber
	\begin{split}
		\mathcal{E}_{i,\Phi_1} := &\left\{{\hat \Sigma}_{p,i,N}  \succcurlyeq \bar \sigma_{p,i}^2 I  \right\}, \\
		\mathcal{E}_{i,\Phi_2} := &\left\{{\hat \Sigma}_{p,i,N}  \preccurlyeq \frac{3d_{p,i}}{\delta} {\Sigma}_{p,i,N}  \right\}, \\
		\mathcal{E}_{i,\Phi_3} := &\left\{ \norm{ \sum_{k=1}^{N}e_i(k)\phi_{p,i}(k)^\top \left(\Sigma + N {\hat \Sigma}_{p,i,N}\right)^{-\frac{1}{2}}}^2 \right. \leq\\ 
		&  \left.4{\text{log}\frac{\text{det}(\Sigma + N {\hat \Sigma}_{p,i,N})}{\text{det}\left(\Sigma\right)}} + 8\left(\text{log}\frac{5^{n_y}3}{\delta} \right) \right\},
	\end{split}
\end{equation}
where matrix $\Sigma \succ 0$. The event $\mathcal{E}_{i,\Phi_1}$ is due to the PE condition in Lemma \ref{Lem1}, the event $\mathcal{E}_{i,\Phi_2}$ is derived from an extension of Markov's inequality in Lemma \ref{LemA4}, and the event $\mathcal{E}_{i,\Phi_3}$ is based on self-normalized martingales in Lemma \ref{LemA5}. 

The main idea is to show that the event in Lemma \ref{Lem2} is subsumed by the intersection of events $\mathcal{E}_{i,\Phi_1}$, $\mathcal{E}_{i,\Phi_2}$ and $\mathcal{E}_{i,\Phi_3}$, with $\mathbb{P}\left(\bigcap_{j=1}^3\mathcal{E}_{i,\Phi_j}\right) \geq 1-\delta$. According to \eqref{E26}, we have that  
\begin{equation}
	\nonumber
	\begin{split}
		\norm{{\tilde \Theta}_{i}^E}^2 &= \norm{\left(H_{fi}\sum_{k=1}^{N}\frac{1}{N}{e_i(k)\phi_{p,i}^\top(k)}{\hat \Sigma}_{p,i,N}^{-\frac{1}{2}}\right){\hat \Sigma}_{p,i,N}^{-\frac{1}{2}}}^2 \\
		&\leq  \underbrace{\norm{H_{fi}\sum_{k=1}^{N}\frac{1}{N}{e_i(k)\phi_{p,i}^\top(k)} {\hat \Sigma}_{p,i,N}^{-\frac{1}{2}}}^2}_{\text{noise term}} \underbrace{\norm{{\hat \Sigma}_{p,i,N}^{-1}}.}_{\text{excitation term}}
	\end{split}	
\end{equation}
The excitation term is bounded based on the event $\mathcal{E}_{i,\Phi_1}$, i.e., 
\begin{equation} \label{BE1}
	\norm{{\hat \Sigma}_{p,i,N}^{-1}} \leq 1/\bar\sigma_{p,i}^2.
\end{equation}
Due to $\mathcal{E}_{i,\Phi_1}$, we have $2N{\hat \Sigma}_{p,i,N} \succcurlyeq N{\hat \Sigma}_{p,i,N} + N\bar\sigma_{p,i}^2I$. After substituting it into the noise term, we have
\begin{equation}
	\nonumber
	\begin{split}
		&\norm{\frac{H_{fi}}{\sqrt{N}}\sum_{k=1}^{N}e_i(k)\phi_{p,i}^\top(k)\left({N\hat \Sigma}_{p,i,N}\right)^{-\frac{1}{2}}}^2\leq\\
		&\frac{2\norm{H_{fi}}^2}{N}\norm{\sum_{k=1}^{N}e_i(k)\phi_{p,i}^\top(k)\left({N\bar\sigma_{p,i}^2I+N\hat \Sigma}_{p,i,N}\right)^{-\frac{1}{2}}}^2.
	\end{split}
\end{equation}
Taking $\Sigma = N\bar\sigma_{p,i}^2I$ in $\mathcal{E}_{i,\Phi_3}$, the noise term is relaxed to
\begin{equation} \label{BE2}
	\begin{split}
		&\norm{\sum_{k=1}^{N}e_i(k)\phi_{p,i}^\top(k)\left({N\bar\sigma_{p,i}^2I+N\hat \Sigma}_{p,i,N}\right)^{-\frac{1}{2}}}^2 \leq \\
		& 4{\text{log}\frac{\text{det}(N\bar\sigma_{p,i}^2I + N {\hat \Sigma}_{p,i,N})}{\text{det}\left(N\bar\sigma_{p,i}^2I\right)}} + 8\left(\text{log}\frac{5^{n_y}3}{\delta} \right) \leq \\
		&4{\text{log}\left(\text{det}(I + \frac{3d_{p,i}}{\delta\bar\sigma_{p,i}^2} \Sigma_{p,i,N})\right)} + 8\left(\text{log}\frac{5^{n_y}3}{\delta} \right)\leq \\
		&4\left(d_{p,i}\text{log}\frac{6d_{p,i}}{\delta} + \text{log}\left(\text{det}(\frac{{\Sigma}_{p,i,N}}{\bar\sigma_{p,i}^2})\right) + 2\left(\text{log}\frac{5^{n_y}3}{\delta} \right)\right),
	\end{split}    	
\end{equation}
where the second inequality is due to $\mathcal{E}_{i,\Phi_2}$, and the third inequality is due to $I\preccurlyeq \frac{3d_i}{\bar\sigma_{p,i}^2\delta} {\Sigma}_{p,i,N}$ under $\mathcal{E}_{i,\Phi_1}$ and $\mathcal{E}_{i,\Phi_2}$. Combining \eqref{BE1} and \eqref{BE2}, and absorbing minor terms by inflating the constants $c_1$ accordingly, we have
\begin{equation}
	\nonumber
	\norm{{\tilde \Theta}_{i}^E}^2 \leq c_1\frac{\epsilon_{i,E}^2}{N},
\end{equation}
where $\epsilon_{i,E}^2 = \frac{\norm{H_{fi}}^2}{\bar\sigma_{p,i}^2}\left(d_{p,i}{\rm{log}}\frac{d_{p,i}}{\delta} + {\rm{log}} \left({\rm{det}}(\frac{{\Sigma}_{p,i,N}}{\bar\sigma_{p,i}^2})\right) \right)$.
\hfill $\blacksquare$
\vspace{-3mm}
\section{Bound on Truncation Bias}  \label{App3}
To prove Lemma \ref{Lem3}, we need the following auxiliary lemma:
\begin{lemma} \label{Lem5}
	Fix a failure probability $0 < \delta < 1$. There exists a universal constant $c_2 \geq 16\sqrt{2}+\frac{1}{{\rm{log}} \frac{2}{\delta}}$, such that with probability at least $1-\delta$,
	\begin{equation} \label{CE1}
		\sum_{k=1}^{N}\norm{x_{k}}^2 \leq c_2 n_x N \norm{\Sigma_{x,N}} {\rm{log}} \frac{2}{\delta}.
	\end{equation}
\end{lemma}
\begin{proof}
	The proof is identical to Lemma E.5 in \cite{Ziemann2023tutorial}, which is an application of Hanson-Wright inequality in Lemma \ref{LemA6}, thus, it is omitted here.
\end{proof}
\vspace{-3mm}
\subsection{Proof of Lemma \ref{Lem3}}
The bias term ${\tilde \Theta}_{i}^B$ can be similarly decomposed as
\begin{equation}
	\nonumber
	{\tilde \Theta}_{i}^B= \left(\Gamma_{fi}A_K^p\sum_{k=1}^{N}{x_{k}\phi_{p,i}^\top(k)}(N\hat \Sigma_{p,i,N})^{-\frac{1}{2}}\right)(N\hat \Sigma_{p,i,N})^{-\frac{1}{2}}.
\end{equation}
Note that ${N\hat \Sigma}_{p,i,N} = \sum_{j=1}^{N}\phi_{p,i}(j)\phi_{p,i}^\top(j) \succcurlyeq \phi_{p,i}(k)\phi_{p,i}^\top(k)$ for each $k$, hence, by the Schur complement, we have
\begin{equation} \label{CE2}
	\phi_{p,i}^\top(k)\left({N\hat \Sigma}_{p,i,N}\right)^{-1}\phi_{p,i}(k) \leq 1.
\end{equation}
Using the triangle inequality, we have
\begin{equation} \label{CE3}
	\begin{split}
		&\norm{\sum_{k=1}^{N}{x_{k}\phi_{p,i}^\top(k)} \left({N\hat \Sigma}_{p,i,N}\right)^{-\frac{1}{2}}} \leq \\ 
		&\sum_{k=1}^{N}\norm{x_{k}} \norm{\phi_{p,i}^\top(k)\left({N\hat \Sigma}_{p,i,N}\right)^{-\frac{1}{2}}}.
	\end{split}	
\end{equation}
Combining \eqref{CE2} with \eqref{CE3}, and using the Cauchy-Schwarz inequality, we further have
\begin{equation} \label{CE4}
	\norm{\sum_{k=1}^{N}{x_{k}\phi_{p,i}^\top(k)} \left({N\hat \Sigma}_{p,i,N}\right)^{-\frac{1}{2}}} \leq	\sqrt{N\sum_{k=1}^{N}\norm{x_{k}}^2}.
\end{equation}
Using inequality \eqref{CE1}, inequality \eqref{CE4} is further relaxed to
\begin{equation} \label{CE5}
	\norm{\sum_{k=1}^{N}{x_{k}\phi_{p,i}^\top(k)} \left({N\hat \Sigma}_{p,i,N}\right)^{-\frac{1}{2}}} \leq \sqrt{c_2 n_x  \norm{\Sigma_{x,N}} N^2{\rm{log}} \frac{1}{\delta}}.
\end{equation}
Combining \eqref{BE1} and \eqref{CE5}, the bias term is bounded by
\begin{equation} \label{CE6}
	\norm{{\tilde \Theta}_{i}^B}^2 \leq \frac{c_2n_x}{\bar\sigma_{p,i}^2} \norm{\Gamma_{fi} A_K^p}\norm{\Sigma_{x,N}}N{\rm{log}} \frac{1}{\delta}.
\end{equation}
Under Assumption \ref{Asp2}, $\norm{\Gamma_{fi}A_K^p}\norm{\Sigma_{x,N}} \leq N^{-3}$, then, the bias error is finally bounded by 
\begin{equation*}
		\norm{{\tilde \Theta}_{i}^B}^2 \leq c_2\frac{\epsilon_{i,B}^2}{N^2},
\end{equation*}
where $\epsilon_{i,B}^2 = \frac{n_x}{{\bar\sigma_{p,i}^2}}{\rm{log}} \frac{1}{\delta}$.
\hfill $\blacksquare$
\vspace{-3mm}
\section{Weighted Singular Value Decomposition}  \label{App4}

To prove Lemma \ref{Lem4} and Theorem \ref{The3}, we first introduce the following auxiliary lemmas:
\begin{lemma} [{\cite[Lemma 5]{Oymak2021revisiting}}] \label{Lem6}
	Fix a failure probability 
	$\delta_u := (2(N+f-1)n_u)^{-{\rm{log}}^2(2fn_u){\rm{log}}(2(N+f-1)n_u)}$.
	There exists a universal constant $c_3$ such that if $N\geq 2c_3fn_y{\rm{log}}(1/\delta_u)$, then with probability at least $1-\delta_u$, we have $\frac{1}{N}U_fU_f^\top \succcurlyeq \frac{1}{2}\sigma_u^2I$.
\end{lemma}

\begin{lemma} \label{Lem7}
	Fix a failure probability $0<\delta<1$. Then, with probability at least $1-3\delta/2$, we have 
	\begin{equation} \label{DE1}
	\bar\sigma_{p,0}^2I \preccurlyeq  \hat \Sigma_{p,0,N} := \frac{1}{N}Z_p Z_p^\top \preccurlyeq \frac{3d_{p,0}}{\delta} {\Sigma}_{p,0,N},
	\end{equation}
	where $\bar\sigma_{p,0} := \frac{\norm{\mathcal{T}_{p,0}}\norm{\Lambda_{u,e}}}{2}>0$, 
	$\mathcal{T}_{p,0} = \begin{bmatrix}G_p&H_p \\I&0 \end{bmatrix}$, $d_{p,0} = p(n_x+n_y)$, $\Sigma_{p,0,k} = \mathbb{E}\left[\phi_{p,0}(k)\phi_{p,0}^\top(k)\right]$, $\phi_{p,0}(k) = \begin{bmatrix}y_p(k)\\ u_p(k)\end{bmatrix}$.
\end{lemma}
\begin{proof}
	The proof for the lower bound is identical to the PE condition in Lemma \ref{Lem1} by taking $i=0$. Meanwhile, the proof for the upper bound is derived from an extension of Markov's inequality in Lemma \ref{LemA4}. For space considertaion, they are omitted here.
\end{proof}
\vspace{-3mm}
\subsection{Proof of Lemma \ref{Lem4} (Properties of Weighting Matrices)}
Since $W_2 = \left(\frac{1}{N}Z_p\Pi_{U_f}^{\perp}Z_p^\top\right)^{\frac{1}{2}}$, it is equivalent to prove that $\frac{1}{N}Z_p\Pi_{U_f}^{\perp}Z_p^\top = \frac{1}{N}Z_pZ_p^\top - \frac{1}{N}Z_pU_f^\top(U_fU_f^\top)^{-1}U_fZ_p^\top$ has the same properties as in Lemma \ref{Lem4}. A prerequisite is that $\frac{1}{N}U_fU_f^\top \succ 0$, which is guaranteed by Lemma \ref{Lem6}.

First, we prove that  $\frac{1}{N}Z_p\Pi_{U_f}^{\perp}Z_p^\top$ is positive definite. According to the second statement on the Schur complement in Lemma \ref{LemA8}, $\frac{1}{N}Z_p\Pi_{U_f}^{\perp}Z_p^\top \succ 0$ is equivalent to requiring
\begin{equation}
	\nonumber
	\frac{1}{N}\begin{bmatrix}Z_pZ_p^\top & Z_pU_f^\top \\U_fZ_p^\top &U_fU_f^\top
	\end{bmatrix} ={\hat \Sigma}_{p,f,N}  \succ 0,
\end{equation}
which is essentially the PE condition in \eqref{E31} by taking $i=f$.

Second, we prove that  $\norm{\frac{1}{N}Z_p\Pi_{U_f}^{\perp}Z_p^\top}$ grows at most logarithmically with $N$. According to Lemma \ref{Lem7}, we have
\begin{equation*}
	\frac{1}{N}Z_p\Pi_{U_f}^{\perp}Z_p^\top \preccurlyeq \frac{1}{N}Z_pZ_p^\top \preccurlyeq \frac{3d_{p,0}}{\delta} {\Sigma}_{p,0,N}.
\end{equation*} 
As a consequence of Lemma \ref{LemA11}, $\norm{\Sigma_{p,0,N}}$ grows at most logarithmically with $N$, we therefore conclude that $\norm{\frac{1}{N}Z_p\Pi_{U_f}^{\perp}Z_p^\top}$ grows at most logarithmically with $N$.

Third, we prove that  $\norm{\left(\frac{1}{N}Z_p\Pi_{U_f}^{\perp}Z_p^\top\right)^{-1}}$ is bounded, which is equivalent to showing that the minimal eigenvalue of $\frac{1}{N}Z_p\Pi_{U_f}^{\perp}Z_p^\top$ is lower bounded. According to the fourth statement of Lemma \ref{LemA8}, we have $\lambda_{\rm{min}}({\hat \Sigma}_{p,f,N}) \leq \lambda_{\rm{min}}(\frac{1}{N}Z_p\Pi_{U_f}^{\perp}Z_p^\top)$. Meanwhile, a lower bound on the minimal eigenvalue of ${\hat \Sigma}_{p,f,N}$ has been given in Lemma \ref{Lem1} by taking $i=f$, thus,  $\norm{\left(\frac{1}{N}Z_p\Pi_{U_f}^{\perp}Z_p^\top\right)^{-1}}$ is bounded. \hfill $\blacksquare$
\vspace{-3mm}
\subsection{Proof of Theorem \ref{The3} (Weighted SVD)}
According to Lemma \ref{LemA9} (let $M = W_1{\mathcal{H}}_{fp}W_2$ and $\hat M = W_1\hat{\mathcal{H}}_{fp}W_2$), if condition \eqref{E41} is satisfied, we have
\begin{equation}
	\nonumber
	\begin{split}
		&\norm{\hat{L}_p- T^\top\bar{L}_p} = \norm{
		\hat \Lambda_{1}^{\frac{1}{2}}\hat V_{1}^\top W_2^{-1}- T^\top\Lambda_{1}^{\frac{1}{2}}V_{1}^\top W_2^{-1}} \\
		\leq & \norm{\hat \Lambda_{1}^{\frac{1}{2}}\hat V_{1}^\top - T^\top\Lambda_{1}^{\frac{1}{2}}V_{1}^\top }\norm{W_2^{-1}}  \\
		\leq & \sqrt{\frac{40n_x}{\sigma_{n_x}(W_1{\mathcal{H}}_{fp}W_2)}}
		\norm{W_1\left(\hat{\mathcal{H}}_{fp}- {\mathcal{H}}_{fp}\right)W_2} \norm{W_2^{-1}} \\
		\leq &
		\sqrt{\frac{40n_x}{\sigma_{n_x}(W_1{\mathcal{H}}_{fp}W_2)}}\norm{\hat{\mathcal{H}}_{fp}- {\mathcal{H}}_{fp}}\norm{W_1}\norm{W_2}\norm{W_2^{-1}} \\
		\leq &\sqrt{\frac{40n_x}{\sigma_{n_x}({\mathcal{H}}_{fp})}} \norm{\hat{\mathcal{H}}_{fp}- {\mathcal{H}}_{fp}}\frac{\norm{W_1}\norm{W_2}\norm{W_2^{-1}}}{\sqrt{\sigma_{\rm {min}} (W_2)\sigma_{\rm {min}} (W_1)}} \\
		= & \sqrt{\frac{40n_x}{\sigma_{n_x}({\mathcal{H}}_{fp})}} \norm{\hat{\mathcal{H}}_{fp}- {\mathcal{H}}_{fp}} W_{L},
	\end{split}
\end{equation}
where $W_{L}$ are given in \eqref{E42}. The first and third inequalities are due to the triangle inequality, the second inequality is due to Lemma \ref{LemA9}, the fourth inequality is due to the relation $\sigma_{n_x} (W_1{\mathcal{H}}_{fp}W_2)\geq \sigma_{\rm {min}} (W_1)\sigma_{n_x} ({\mathcal{H}}_{fp})\sigma_{\rm {min}} (W_2)$, where the rank of ${\mathcal{H}}_{fp}$ is $n_x$, and the last equality is due to $\sigma_{\rm {min}} (W_2)\sigma_{\rm {min}} (W_1) = \norm{W_2^{-1}}^{-1}\norm{W_1^{-1}}^{-1}$. The bound of $\norm{\hat{\Gamma}_{f}- \bar{\Gamma}_{f}T}$ can be similarly derived.
\hfill $\blacksquare$
\vspace{-3mm}
\section{Bounds on System matrices}   \label{App5}
\subsection{Proof of Theorem \ref{The4} (CVA Type Realization)}
\subsubsection{Bound on $C$}
To estimate $C$, we first rewrite \eqref{E19a} as
\begin{equation*}
	\begin{split}
		Y_{f1} &= {\bar C T }{T^\top \bar X_k} + \Sigma_e^{1/2}E_{f1} \\
		&= {\bar C T }\hat X_k + {\bar C T }(T^\top \bar X_k-\hat X_k) + \Sigma_e^{1/2}E_{f1}.
	\end{split}	
\end{equation*}
According to \eqref{E20a}, the estimate of $C$ is $\hat C = Y_{f1}\hat X_k^\dagger$, where $\hat X_k^\dagger = Z_p^\top\hat L_p^\top(\hat L_p Z_pZ_p^\top\hat L_p^\top)^{-1}$. In this way, the estimation error of $C$ can be decomposed into three types:
\begin{equation} \label{EE1}
	\begin{split}
		\tilde{C} :=& \hat C - {\bar C T } = {\bar C T }(T^\top \bar X_k-\hat X_k)\hat X_k^\dagger + \Sigma_e^{1/2}E_{f1}\hat X_k^\dagger\\
		=& \tilde{C}^L + \tilde{C}^B + \tilde{C}^E,
	\end{split}
\end{equation}
where 
\begin{equation*}
	\begin{split}
		\tilde{C}^L &:= {\bar C T }(T^\top \bar L_p- \hat{L}_p)Z_p\hat X_k^\dagger, \\ \tilde{C}^B &:= {C}A_K^pX_{k-p}\hat X_k^\dagger, \\
		\tilde{C}^E &:= \Sigma_e^{1/2}E_{f1}\hat X_k^\dagger.
	\end{split}
\end{equation*}

First, we rewrite the error coming from $\hat{L}_p$ as
\begin{equation*}
	\tilde{C}^L = {\bar C T }(T^\top \bar L_p- \hat{L}_p)Z_pZ_p^\top\hat L_p^\top(\hat L_p Z_pZ_p^\top\hat L_p^\top)^{-1}.
\end{equation*}
For convenience, we define 
\begin{equation*}
	\text{Cond}(\Phi_{p,0}\Phi_{p,0}) = \hat L_pZ_pZ_p^\top\hat L_p^\top(\hat L_p Z_pZ_p^\top\hat L_p^\top)^{-1},
\end{equation*}
where $\Phi_{p,0}=Z_p$. Using Lemma \ref{Lem7}, we have that
\begin{equation*}
	\norm{\text{Cond}(\Phi_{p,0}\Phi_{p,0})} \leq \frac{\norm{\hat L_p}^2}{\sigma_{\min}^2(\hat L_p)}\frac{\norm{\Sigma_{p,0,N}}}{\bar\sigma_{p,0}^2}
	\frac{3d_{p,0}}{\delta},
\end{equation*}
which grows at most logarithmically with $N$. In this way, we have that 
\begin{equation}  \label{EE2}
    \norm{\tilde{C}^L} \leq c_4 \norm{\hat{L}_p- T^\top \bar L_p},
\end{equation}
where $c_4 := \norm{\bar C T} \norm{\text{Cond}(\Phi_{p,0}\Phi_{p,0})}$.

Second, the truncation bias can be rewritten as
\begin{equation}
	\nonumber
	\begin{split}
	 \tilde{C}^B &= CA_K^pX_{k-p}Z_p^\top\hat L_p^\top(\hat L_p Z_pZ_p^\top\hat L_p^\top)^{-1}\\
	& = CA_K^pX_{k-p}Z_p^\top(Z_pZ_p^\top)^{-1}\text{Cond}(\Phi_{p,0}\Phi_{p,0}).
	\end{split}
\end{equation}
Similar to bounding ${\tilde \Theta}_{i}^B$ in Lemma \ref{Lem3}, by taking $i=0$, $\tilde{C}^B$ is bounded by 
\begin{equation}  \label{EE3}
	\norm{\tilde{C}^B}^2 
	\leq \frac{c_5\epsilon_{0,B}^2}{N^2},
\end{equation}
where $\epsilon_{0,B}^2 := \frac{n_x}{{\bar\sigma_{p,0}^2}}{\rm{log}}\frac{1}{\delta}
$ and $c_5 := c_2\norm{C}\norm{\text{Cond}(\Phi_{p,0}\Phi_{p,0})}$.

Third, the cross-term error can be rewritten as
\begin{equation}
	\nonumber
	\begin{split}
		\tilde{C}^E &= \Sigma_e^{1/2}E_{f1}Z_p^\top\hat L_p^\top(\hat L_p Z_pZ_p^\top\hat L_p^\top)^{-1}\\
		& = \Sigma_e^{1/2}E_{f1}Z_p^\top(Z_pZ_p^\top)^{-1}\text{Cond}(\Phi_{p,0}\Phi_{p,0}).
	\end{split}
\end{equation}
Similar to bounding ${\tilde \Theta_{i}}^E$ in Lemma \ref{Lem2}, by taking $i=0$, $\tilde{C}^E$ is bounded by 
\begin{equation} \label{EE4}
	 \norm{\tilde{C}^E} \leq \frac{c_6\epsilon_{0,E}^2}{N},
\end{equation}
where $c_6 := c_1\norm{\text{Cond}(\Phi_{p,0}\Phi_{p,0})}$, and $\epsilon_{0,E}$ is given by taking $i=0$ in $\epsilon_{i,E}$ in \eqref{E32}. After merging \eqref{EE2}, \eqref{EE3} and \eqref{EE4} together, we obtain \eqref{E43a}.

\subsubsection{Bounds on $A$ and $B$}
We first rewrite \eqref{E19b} as
\begin{equation}
	\nonumber
 \begin{split}
     \hat X_k^{+} = &{T^\top \bar A T}\hat X_k+ {T^\top \bar B}U_{f1} + K\Sigma_e^{1/2}E_{f1} +  \\
     &{T^\top \bar A T}(T^\top \bar X_k-\hat X_k) +  (\hat X_k^{+} - T^\top \bar X_k^{+}).
 \end{split}
\end{equation}
According to \eqref{E20b}, the estimates of $A$ and $B$ is
\begin{equation} \label{EE5}
	\hat \theta := \begin{bmatrix}\hat A& \hat B \end{bmatrix} = \hat X_k^{+} \begin{bmatrix}\hat X_k\\U_{f1}\end{bmatrix}^\dagger.
\end{equation}
For brevity, we define 
\begin{equation*}    
    \bar\theta := \begin{bmatrix}T^\top \bar A T&T^\top \bar B \end{bmatrix},
    \Psi_k := \begin{bmatrix}\hat X_k\\U_{f1}\end{bmatrix} = \hat L_{p,1}\Phi_{p,1}, 
\end{equation*}	
where $\hat L_{p,1} = \begin{bmatrix}\hat L_p&0\\0&I\end{bmatrix}$, and $\Phi_{p,1}=\begin{bmatrix}Z_p\\U_{f1}\end{bmatrix}$. In this way, the estimation error of $\theta$ can be similarly divided into three parts:
\begin{equation*}
	\tilde{\theta} := \hat \theta-\bar\theta =  \tilde{\theta}^L + \tilde{\theta}^B + \tilde{\theta}^E,
\end{equation*}
where
\begin{equation*}
	\begin{split}
		\tilde{\theta}^L &:= T^\top \bar A T(\hat{L}_p- T^\top \bar L_p)Z_p\Psi_{k}^\dagger +  (\hat{L}_p- T^\top \bar L_p)Z_p^{+}\Psi_{k}^\dagger, \\ 
		\tilde{\theta}^B &:= AA_K^pX_{k-p}\Psi_{k}^\dagger +  A_K^pX_{k-p}^{+}\Psi_{k}^\dagger, \\ 
		\tilde{\theta}^E &:= K\Sigma_e^{1/2}E_{f1}\Psi_{k}^\dagger.
	\end{split}
\end{equation*}

First, for $\tilde{\theta}^L$, we rewrite it as
\begin{equation*}
	\begin{split}
		\tilde{\theta}^L :=& T^\top \bar A T(\hat{L}_p- T^\top \bar L_p)Z_p\Phi_{p,1}^\top\hat{L}_{p,1}^\top\left(\hat{L}_{p,1}\Phi_{p,1}\Phi_{p,1}^\top\hat{L}_{p,1}^\top\right)^{-1}\\
		& + (\hat{L}_p- T^\top \bar L_p)Z_p^{+}\Phi_{p,1}^\top\hat{L}_{p,1}^\top\left(\hat{L}_{p,1}\Phi_{p,1}\Phi_{p,1}^\top\hat{L}_{p,1}^\top\right)^{-1}.
	\end{split}	
\end{equation*}
For convenience, we define 
\begin{equation*}
	\text{Cond}(\Phi_{p,1}\Phi_{p,1}^\top) = \hat{L}_{p,1}\Phi_{p,1}\Phi_{p,1}^\top\hat{L}_{p,1}^\top\left(\hat{L}_{p,1}\Phi_{p,1}\Phi_{p,1}^\top\hat{L}_{p,1}^\top\right)^{-1}.
\end{equation*}
Similar to $\text{Cond}(\Phi_{p,0}\Phi_{p,0})$, we conclude that $\text{Cond}(\Phi_{p,1}\Phi_{p,1})$ grows at most logarithmically with $N$. Then, we have that
\begin{equation} \label{EE6}
	\norm{\tilde{\theta}^L} \leq c_7 \norm{\hat{L}_p- T^\top \bar L_p},
\end{equation}
where $c_7 := \left(\norm{T^\top \bar A T}+\norm{I}\right) \norm{\text{Cond}(\Phi_{p,1}\Phi_{p,1}^\top)}$.

Second, for $\tilde{\theta}^B$, we rewrite it as
\begin{equation*}
	\begin{split}
		\tilde{\theta}^B = &AA_K^pX_{k-p}\Phi_{p,1}^\top\hat{L}_{p,1}^\top\left(\hat{L}_{p,1}\Phi_{p,1}\Phi_{p,1}^\top\hat{L}_{p,1}^\top\right)^{-1} +  \\
		&A_K^pX_{k-p}^{+}\Phi_{p,1}^\top\hat{L}_{p,1}^\top\left(\hat{L}_{p,1}\Phi_{p,1}\Phi_{p,1}^\top\hat{L}_{p,1}^\top\right)^{-1},
	\end{split}	
\end{equation*}
which can be bounded by
\begin{equation} \label{EE7}
	\norm{\tilde{\theta}^B} \leq c_8 \frac{\epsilon_{1,B}}{N},
\end{equation}
where $c_8 = \norm{A+I}\norm{\text{Cond}(\Phi_{p,1}\Phi_{p,1}^\top)}$, and $\epsilon_{1,B}$ is in \eqref{E34}.

Third, for $\norm{{\tilde C}^E}$, we rewrite it as
\begin{equation*}
	\tilde{\theta}^E := K\Sigma_e^{1/2}E_{f1}\Phi_{p,1}^\top\hat{L}_{p,1}^\top\left(\hat{L}_{p,1}\Phi_{p,1}\Phi_{p,1}^\top\hat{L}_{p,1}^\top\right)^{-1},
\end{equation*}
which can be bounded by
\begin{equation} \label{EE8}
	\norm{\tilde{\theta}^E} \leq c_9 \frac{\epsilon_{1,E}}{\sqrt{N}},
\end{equation}
where $c_9 = \norm{K\Sigma_e^{1/2}}\norm{\text{Cond}(\Phi_{p,1}\Phi_{p,1}^\top)}$, and $\epsilon_{1,E}$ is in \eqref{E32}. After merging \eqref{EE6}, \eqref{EE7} and \eqref{EE8} together, we obtain \eqref{E43b}. We hence complete the proof. \hfill $\blacksquare$
\vspace{-3mm}
\subsection{Proof of Theorem \ref{The5}  (MOESP Type Realization)}

The proof of Theorem \ref{The5} is identical to \cite[Th. 4]{Tsiamis2019finite}, thus, it is omitted here. \hfill $\blacksquare$
\vspace{-3mm}
\section{Technical Lemmas}   \label{App6}

\begin{lemma} [{\cite[Lemma 6]{Abbasi2011regret}}] \label{LemA1}	
	Norm of a sub-Gaussian vector: For an entry-wise $\sigma_w^2$-sub-Gaussian random variable $w\in \mathbb{R}^{n_w}$, i.e., such that ${\rm{log}}\mathbb{E}\left[e^{\lambda w}\right] \leq \mathbb{E}\left[w\right]\lambda + \frac{\sigma_w^2\lambda2}{2}$ for all $\lambda \in \mathbb{R}$, with probability at least $1-\delta/2$, for $1\leq k \leq N$,
	\begin{equation}
		\norm{w}\leq \sigma_w n_w\sqrt{2n_w{\rm {log}}(32Nn_w/\delta)}.
		\nonumber
	\end{equation}
\end{lemma}

% \begin{lemma} [Lemma 14 in \cite{Abbasi2011regret}] \label{LemA2-1}	
%     Let $X_1, \dots, X_t$ be random variables. Let $a \in \mathbb{R}$. Let $S_t = \sum_{s=1}^{t} X_s$ and $\widetilde{S}_t = \sum_{s=1}^{t} \mathbf{1}_{X_s \leq a} X_s$ where $\mathbf{1}_{X_s \leq a} X_s$ denotes the truncated version of $X_s$. Then it holds that
%      \begin{equation}
%     \mathbb{P}(S_t > x) \leq \mathbb{P}\left(\max_{1 \leq s \leq t} X_s \geq a\right) + \mathbb{P}(\widetilde{S}_t > x).
%      \nonumber
%     \end{equation}
% \end{lemma}

\begin{lemma}  \label{LemA2-1}	
    Let $M_1, \dots, M_N$ be random matrices in $\mathbb{R}^{n\times n}$. Let $W \in \mathbb{R}^{n\times n}$ be a fixed symmetric matrix. Let $S_N = \sum_{i=1}^{N} M_i$ and $\tilde{S}_N = \sum_{i=1}^{N} \tilde{M}_i$, where $\tilde{M}_i = M_i\mathbb{I}_{\left\{{M_i \preccurlyeq W}\right\}}$ is the truncated version of $M_i$. Then it holds that
     \begin{equation*}
         \mathbb{P}\left(S_N \succ V\right) \leq \mathbb{P}\left(\max\limits_{1\leq i\leq N} M_{i} \succ W \right) + \mathbb{P}\left(\tilde{S}_N \succ V\right).
     \end{equation*}
\end{lemma}
\begin{proof}
The proof closely follows the proof of the scalar case \cite[Lemma 14]{Abbasi2011regret}, thus, it is omitted here.
% We have
%      \begin{equation*}
%      \begin{split}
%          &\mathbb{P}\left(S_N \succ V\right) = \mathbb{P}\left(S_N \succ V, \bigcup_{i=1}^{N} M_{i} \succ W\right) + \\
%          &\mathbb{P}\left(S_N \succ V, {\left(\bigcup_{i=1}^{N} M_{i} \succ W\right)}^c\right) 
%          \leq \\
%          &\mathbb{P}\left(\bigcup_{i=1}^{N} M_{i} \succ W\right) + 
%          \mathbb{P}\left(S_N \succ V, \bigcap_{i=1}^{N} M_{i} \preccurlyeq W\right) = \\ &\mathbb{P}\left(\bigcup_{i=1}^{N} M_{i} \succ W\right) + \mathbb{P}\left(\tilde{S}_N \succ V\right),
%      \end{split}    
%      \end{equation*}
%      where the last inequality is due to $M_i = \tilde{M}_i$ on the event $\bigcup_{i=1}^{N} M_{i} \preccurlyeq W$.
\end{proof}

\begin{lemma} [{\cite[Th. 7.1]{Tropp2012user}}]\label{LemA2}
	Matrix Azuma: Consider a finite adapted sequence $\left\{W_k\right\}$ of self-adjoint matrices of dimension $d$, and a fixed sequence $\left\{M_k\right\}$ of self-adjoint matrices that satisfy $\mathbb{E}_{k-1}W_k = 0$ and $M_k^2\succcurlyeq W_k^2$ almost surely. Then, for all $t\geq 0$, $\mathbb{P}\left\{\lambda_{\rm{max}}(\sum_{k} W_k)\geq t\right\} \leq d\cdot{\rm{exp}}\left({-\frac{t^2}{8\norm{\sum_{k} M_k^2}}}\right)$.
	
%	\begin{equation}
%		\mathbb{P}\left\{\lambda_{\rm{max}}(\sum_{k} W_k)\geq t\right\} \leq d\cdot{\rm{exp}}\left({-\frac{t^2}{8\norm{\sum_{k} M_k^2}}}\right).
%		\nonumber
%	\end{equation}
\end{lemma}

% \begin{lemma} [Corollary 7.2 in \cite{Tropp2012user}]\label{LemA2-1}
% 	Consider an self-adjoint matrix martingale $\left\{Y_k: k =1,...,n\right\}$ in dimension $d$, and let $\left\{X_k \triangleq Y_k - Y_{k-1}\right\}$ be the associated difference sequence. Suppose that the difference sequence satisfies the hypothesis of Lemma \ref{LemA2}, then for all $t\geq 0$,
% 	\begin{equation}
% 		\mathbb{P}\left\{\lambda_{\rm{max}}(Y_n - \mathbb{E}Y_n)\geq t\right\} \leq d\cdot{\rm{exp}}\left({-\frac{t^2}{8\norm{\sum_{k} M_k^2}}}\right).
% 		\nonumber
% 	\end{equation}
% \end{lemma}

\begin{lemma} [\cite{Tao2010254a}]\label{LemA3}
	Weyl's inequality: Let $M_1, M_2 \in \mathbb{R}^{n\times n}$ be Hermitian matrices, with eigenvalues ordered in descending order $\lambda_1\geq \lambda_2\geq \cdots \geq \lambda_n$. Then, for $i+j > n$,
	\begin{equation*}
		\lambda_{i+j-1}(M_1+M_2) \leq \lambda_{i}(M_1)+\lambda_{j}(M_2) \leq \lambda_{i+j-n}(M_1+M_2).
	\end{equation*}
\end{lemma}

\begin{lemma} [{\cite[Th. 12]{Ahlswede2002strong}}] \label{LemA4}
	Markov's inequality: Let a matrix $M \succ 0$, and a random matrix $W \succcurlyeq 0$ almost surely. Then, we have $\mathbb{P}\left(W\not\preceq M\right) \leq {\rm {trace}}\left(\mathbb{E}W M^{-1}\right)$, where $\left(W\not\preceq M\right)$ is the complement of the event $\left(W\preceq M\right)$.
\end{lemma}

\begin{lemma} [{\cite[Th. 3.4]{Abbasi2011online}}]\label{LemA5}
	  Self-normalized martingale: Let $\left\{\mathcal{F}_k\right\}_{k=0}^N$ be a filtration such that $\left\{W_k\right\}_{k=1}^N$ is adapted to $\left\{\mathcal{F}_{k-1}\right\}_{k=1}^N$ and $\left\{V_k\right\}_{k=1}^N$ is adapted to $\left\{\mathcal{F}_{k}\right\}_{k=1}^N$. Additionally, suppose that for all $1\leq k \leq N$, $V_k \in \mathbb{R}^{n_v\times 1}$ is $\sigma^2-$conditionally sub-Gaussian with respect to $\mathcal{F}_{k}$. Let $\Sigma \in \mathbb{R}^{n_w\times n_w}$. Given a failure probability $0<\delta<1$, then with probability at least $1-\delta$, we have 
	  \begin{equation*}
	  	\begin{split}
	  		&\norm{\sum_{k=1}^{N}V_kW_k^\top \left(\Sigma+\sum_{k=1}^{N}W_kW_k^\top\right)^{-1/2}}^2 \leq 8n_v\sigma^2{\rm {log}}5 + \\
	  		&4\sigma^2{\rm {log}}\left(\frac{{\rm {det}}\left(\Sigma+\sum_{k=1}^{N}W_kW_k^\top\right)}{{\rm {det}}\left(\Sigma\right)}\right) +  8\sigma^2{\rm {log}}\frac{1}{\delta}.
	  	\end{split}	  	
	  \end{equation*}
\end{lemma}

\begin{lemma} [{\cite{Hanson1971bound}, \cite[Th. 2.1]{Ziemann2023tutorial}}]\label{LemA6}
	 Hanson-Wright inequality: Consider a random variable $w\in\mathbb{R}^{n_w}$, where each entry is a scalar, zero mean and independent $\sigma_w^2$-sub-Gaussian random variable. For a matrix $M\in\mathbb{R}^{n_w\times n_w}$ and every $s\geq0$, we have
	\begin{equation}
		\nonumber
		\begin{split}
			&\mathbb{P}\left(\norm{w^\top Mw - \mathbb{E}w^\top Mw}>s\right) \leq \\
			&2{\rm {exp}}\left(-{\rm {min}}\left(\frac{s^2}{114\sigma_w^4\norm{M}_F^2},\frac{s}{16\sqrt{2}\sigma_w^2\norm{M}}\right)\right).
		\end{split}
	\end{equation}
\end{lemma}

\begin{lemma} [Lemma A.1 in \cite{Tsiamis2019finite}] \label{LemA7}
	Norm of a block matrix: Let a block matrix $M = \begin{bmatrix}
		M_1^\top&M_2^\top&\cdots &M_f^\top
	\end{bmatrix}^\top$, where all the sub-matrices $M_i$'s have the same dimension. Then, the block matrix $M$ satisfies $\norm{M} \leq \sqrt{f} \max\limits_{1\leq i\leq f} \norm{M_i}$.
%	\begin{equation}
%		\nonumber
%		\norm{M} \leq \sqrt{f} \max\limits_{1\leq i\leq f} \norm{M_i}.
%	\end{equation}
\end{lemma}

\begin{lemma} \label{LemA8}
	Shur complement: Let $M = \begin{bmatrix}
		M_1&M_2\\M_2^\top&M_4
	\end{bmatrix}$ be a block matrix, and $M_4 \succ 0$. Defining $M_S = M_1-M_2M_4^{-1}M_2^\top$ as the Shur complement of $M$, we then have 
\begin{enumerate}
	\item $M^{-1} = \begin{bmatrix}M_S^{-1}&-M_S^{-1}M_2M_4^{-1}\\-M_4^{-1}M_2^\top M_S^{-1}&M_\Delta\end{bmatrix}$, where $M_\Delta = M_4^{-1}+M_4^{-1}M_2^\top M_S^{-1}M_2M_4^{-1}$.
	\item $M \succ 0$, if and only if $M_S\succ 0$.
	\item If $M \succ 0$, then $\lambda_{\rm{max}}(M) \geq \lambda_{\rm{max}}(M_S)$.
	\item If $M \succ 0$, then $\lambda_{\rm{min}}(M) \leq \lambda_{\rm{min}}(M_S)$.
\end{enumerate}
\end{lemma}
\begin{proof}
Given that $M_4 \succ 0$, then $M$ can be rewritten as
\begin{equation}
	\nonumber
	M = \begin{bmatrix}I&M_2M_4^{-1}\\0&I\end{bmatrix} 
\begin{bmatrix}M_S&0\\0&M_4\end{bmatrix}
\begin{bmatrix}I&0\\M_4^{-1}M_2^\top&I\end{bmatrix}. 
\end{equation}

The first and second statements can then be obtained straightforwardly. For the third statement, since $M_S \prec M_1$, we have $\lambda_{\rm{max}}(M_S) \leq \lambda_{\rm{max}}(M_1) \leq \lambda_{\rm{max}}(M)$. For the forth statement, according to the first statement, since $\lambda_{\rm{max}}(M^{-1}) \geq \lambda_{\rm{max}}(M_S^{-1})$, so $\lambda_{\rm{min}}(M) \leq \lambda_{\rm{min}}(M_S)$.

% first we define $\mathcal{M} \triangleq \left\{\begin{bmatrix}m\\n\end{bmatrix}: M_4^{-1}M_2^\top m+n = 0\right\}$. Note that if $z =\begin{bmatrix}m\\n\end{bmatrix} \in \mathcal{M}$, then $\begin{bmatrix}I&0\\M_4^{-1}M_2^\top&I\end{bmatrix}\begin{bmatrix}m\\n\end{bmatrix}=\begin{bmatrix}m\\0\end{bmatrix}$. In this way, we have
% \begin{equation}
% 	\nonumber
% 	\begin{split}
% 		\lambda_{\rm{min}}(M) &= \min\limits_{z\neq 0} \frac{z^\top Mz}{z^\top z} \leq \min\limits_{0\neq z\in\mathcal{M}} \frac{z^\top Mz}{z^\top z} \\
% 		&\leq \min\limits_{m\neq 0} \frac{m^\top M_Sm}{m^\top m} = \lambda_{\rm{min}}(M_S). 
% 	\end{split}
% \end{equation}
\end{proof}

\begin{lemma} [{\cite[Th. 4]{Tsiamis2019finite}}] \label{LemA9}
	Suppose rank $n$ matrices $M$ and $\bar M$ have singular value decomposition $U\Lambda V^\top$ and $\bar U \bar \Lambda \bar V^\top$, where $\bar M$ is the rank $n$ approximation of $\hat M$. If $\norm{M- \hat M} \leq \frac{\sigma_n (M)}{4}$, then there exists a unitary matrix $T$ such that
	\begin{equation}
		\nonumber
		\max\left(\norm{\bar U \bar \Lambda^{1/2} - U  \Lambda^{1/2}T},\norm{\bar \Lambda^{1/2}\bar V^\top -  T^\top \Lambda^{1/2}V^\top}\right) \leq \kappa_M,
	\end{equation}
    where $\kappa_M = \sqrt{\frac{40n}{\sigma_n (M)}}\norm{M- \hat M}$.
\end{lemma}

\begin{lemma} [{\cite[Th. 4.1]{Wedin1973perturbation}}] \label{LemA10}
	Consider matrices $M_1, M_2\in \mathbb{R}^{m\times n}$ with rank $m$, where $m\leq n$. Then, we have 
	\begin{equation}
		\nonumber
		\norm{M_1^\dagger-M_2^\dagger} \leq \sqrt{2}\norm{M_1^\dagger}\norm{M_2^\dagger}\norm{M_1-M_2}.
	\end{equation}
\end{lemma}

\begin{lemma} [{\cite[Lemma E.2 ]{Tsiamis2019finite}}] \label{LemA11}
	Consider the series $W_t = \sum_{k=0}^{t} \norm{A^k}$. We have the following two cases:
	\begin{enumerate}
	\item If the system is asymptotically stable, i.e., $\rho(A) < 1$, then $\norm{W_t} = \mathcal{O}(1)$.
	\item If the system is marginally stable, i.e., $\rho(A) = 1$, then $\norm{W_t} = \mathcal{O}(t^\kappa)$, where $\kappa$ is the Largest Jordan blcok of $A$ corresponding to the eigenvalue $|\lambda|=1$.
    \end{enumerate}
\end{lemma}

\begin{lemma} [{\cite[Lemma B.5]{Tsiamis2019finite}}] \label{LemA12}
	Let $\beta \geq b > 0$, for some $b>0$ and consider the function:
	\begin{equation*}
		f(\alpha,\beta) = \frac{\alpha+\beta}{2} -\gamma\sqrt{\alpha+1}, 
	\end{equation*}
    for $a\geq 0$. If $\gamma\leq \min\left(1,\frac{b}{2}\right)$, then 
    $f(\alpha,\gamma) \geq 0$.
\end{lemma}

\end{document}